\documentclass{acm_proc_article-sp}

\usepackage{rotating}

\def\endproof{\hfill$\Box$\vspace{0.4cm}}
\newtheorem{propo}{Proposition}[section]
\newtheorem{lemma}[propo]{Lemma}

\newtheorem{thm}[propo]{Theorem}
\newtheorem{theorem}[propo]{Theorem}

\newtheorem{remark}[propo]{Remark}

\newcommand{\stat}{\mu}
\newcommand{\<}{\langle}
\renewcommand{\>}{\rangle}

\newcommand{\reals}{{\mathds R}}

\newcommand{\naturals}{{\mathds N}}

\def\l|{\left|\left|}
\def\r|{\right|\right|}

\def\P{\mathbb P}

\def\E{\mathbb E}
\def\prob{{\mathbb P}}

\def\ind{{\mathbb I}}

\def\tu{\widetilde{u}}
\def\hu{\hat{u}}
\def\hU{\widehat{U}}
\def\hu{\widehat{u}}
\def\tx{\widetilde{x}}

\def\Trace{{\rm Tr}}
\def\cA{{\cal A}}
\def\cB{{\cal B}}
\def\cN{{\sf N}}

\def\cP{{\cal P}}

\def\tlam{\widetilde{\lambda}}

\newcommand{\srank}{{\gamma}}
\newcommand{\Ms}{S}
\newcommand{\Mt[1]}{S^{(#1)}}

\newcommand{\hlt[1]}{{\lambda}^{(#1)}}
\newcommand{\var}{\text{\rm Var}}
\newcommand{\Var}{\text{\rm Var}}

\newcommand{\bE}{\mathbb{E}}

\newcommand{\good}{\mathcal{G}}
\newcommand{\nogood}{\overline{\mathcal{G}}}
\newcommand{\cF}{\mathcal{F}}

\renewcommand{\theenumi}{{\roman{enumi}}}

\def\ve{\varepsilon}
\def\prob{{\mathbb P}}
\def\hlam{\widehat{\lambda}}
\def\ind{{\mathbb I}}

\def\id{{\mathbb I}}

\def\reals{{\mathbb R}}
\def\naturals{{\mathbb N}}
\def\<{\langle}
\def\>{\rangle}
\def\E{{\mathbb E}}
\def\P{{\sf P}}
\def\de{{\rm d}}

\def\bu{{\bf u}}

\def\bhU{{\bf \widehat{U}}}
\def\bx{{\bf x}}
\def\by{{\bf y}}
\def\bz{{\bf z}}
\def\bU{{\bf U}}
\def\bX{{\bf X}}
\def\bY{{\bf Y}}
\def\bZ{{\bf Z}}
\def\tlam{{\widetilde{\lambda}}}

\def\LA{{\sf L1}}
\def\LB{{\sf L2}}
\def\Aa{{\sf A1}}
\def\Ab{{\sf A2}}
\def\Abb{{\sf A2'}}

\def\G{{\sf G}}

\def\cF{{\cal F}}

\def\normal{{\sf N}} 
\def\gossip_pca{{\sc Gossip PCA}}

\begin{document}

\title{Gossip PCA}

\numberofauthors{3} %  in this sample file, there are a *total*
\author{
% 1st. author
\alignauthor
	Satish Babu Korada\\
       \affaddr{Electrical Engineering Department}\\
       \affaddr{Stanford, CA 94305}\\
       \email{satishbabu.k@gmail.com}
% 2nd. author
\alignauthor
	Andrea Montanari\\
       \affaddr{Electrical Engineering and Statistics Departments}\\
       \affaddr{Stanford University}\\
       \affaddr{Stanford, CA 94305}\\
       \email{montanari@stanford.edu}
% 3rd. author
\alignauthor
	Sewoong Oh\\
       \affaddr{EECS Department}\\
       \affaddr{Massachusetts Institute of Technology}\\
       \affaddr{Cambridge, MA 02139}\\
       \email{swoh@mit.edu}
}
% There's nothing stopping you putting the seventh, eighth, etc.
% author on the opening page (as the 'third row') but we ask,
% for aesthetic reasons that you place these 'additional authors'
% in the \additional authors block, viz.
% Just remember to make sure that the TOTAL number of authors
% is the number that will appear on the first page PLUS the
% number that will appear in the \additionalauthors section.

\maketitle

\begin{abstract}

Eigenvectors of data matrices play an important role in many computational
problems, ranging from signal processing to machine learning and control.
For instance, algorithms that compute positions of the nodes of a
wireless network on the basis of pairwise distance measurements require 
a few leading eigenvectors of the distances matrix.
While eigenvector calculation is a standard topic in numerical linear
algebra, it becomes challenging under severe communication or 
computation constraints, or in absence of central scheduling.
In this paper we investigate the possibility of computing the leading
eigenvectors of a large data matrix through gossip algorithms.

The proposed algorithm amounts to iteratively multiplying 
a vector by independent random
sparsification of the original matrix and averaging the resulting
normalized vectors. This can be viewed as a generalization of
gossip algorithms for consensus, but the resulting dynamics is significantly 
more intricate. Our analysis is based on controlling the
convergence to stationarity of the associated Kesten-Furstenberg Markov chain.

\end{abstract}



%
%=========================================================================
%
\section{Introduction and overview}
\label{sec:intro}

Consider a system formed by $n$ nodes with limited
computation and communication capabilities, and connected 
via the complete graph $K_n$. To each edge $(i,j)$
of the graph is associated the entry $M_{ij}$ of an $n\times n$
symmetric matrix $M$. Node $i$ has access to the entries
of $M_{ij}$ for $j\in\{1,\dots,n\}$. An algorithm is required to compute 
the eigenvector of $M$ corresponding to the eigenvalue 
with the largest magnitude. Denoting by $u\in\reals^n$ the eigenvector,
each node $i$ has to compute the corresponding entry $u_i$.
The eigenvector $u$ is often called the \emph{principal component} of $M$,
and analysis methods that approximate a data matrix by its 
leading eigenvectors are referred to as principal component
analysis \cite{Jolliffe}.

Eigenvector calculation is a key step in many computational 
tasks, e.g. dimensionality reduction \cite{Isomap}, classification \cite{StatisticalLearning}, 
latent semantic indexing \cite{LSI}, link analysis (as in PageRank) \cite{PageRank}.
The primitive developed in this paper can therefore be useful
whenever such tasks have to be performed under stringent 
communication and computation constraints. As a stylized 
application, consider the case in which the nodes are $n$
wireless hand-held devices (for related commercial products,
see \cite{Ekahau,Qwikker,Sonitor}). Accurate positioning of the nodes
in indoor environments is difficult through standard methods such as GPS
\cite{GPS}.
Because of intrinsic limitation of GPS and of roof scattering,
indoor position uncertainty can be of 10 meters or larger, which
is too much for locating a room in a building.
An alternative approach consists in measuring pairwise distances 
through delay measurements between the nodes
and reconstructing the nodes positions 
from such measurements (obviously this is possible only up
to a global rotation or translation).
Positions indeed 
can be extracted from the matrix of square distances by computing
its three leading eigenvectors (after appropriate centering) \cite{OKM10}.
This method is known as multidimensional scaling, and 
we will use it as a running example throughout this paper.

A simple centralized method for computing the eigenvector
is to collect all the matrix entries at one special node, say node $i$, 
to perform the eigenvector calculation there and then 
flood back its entries to each node. This centralized approach 
has several disadvantages. It requires communicating $n^2$
real numbers through the network at the beginning
of execution, and puts
a large memory, computation and communication burden on node $i$.
It is also very fragile to failure or Byzantine behavior of $i$.

The next simplest idea is to use some version of the \emph{power} \emph{method}. 
A decentralized power method would proceed by 
synchronized iterations through the network. 
At $t$-th iteration, each node keeps a running estimate $x^{(t)}_i$
of the leading eigenvector. This is updated by letting 
$x^{(t+1)}_i = \sum_{j=1}^nM_{ij}x^{(t)}_j$. If $M$ has strong spectral 
features
(in particular, if the two largest eigenvalues are not close) these estimates
will converge rapidly. On the other hand, each iteration requires
$(n-1)$ real numbers to be transmitted to each node, and 
$n$ sums and multiplications to be performed at the node. In other words, 
the node capabilities have to scale with the network size. This problem
becomes even more severe for wireless devices, which are intrinsically
interference-limited. Within the power method approach, $n^2$ communications 
have to be scheduled at each time thus requiring significant bandwidth.
Finally, the algorithm requires complete synchronization of the $n^2$
communications and is fragile to link failures (which can be quite frequent
e.g. due to fading).

A simple and yet powerful idea that overcomes some of these problems is 
sparsification. Throughout the paper, we say that $S\in\reals^{n\times n}$
is a \emph{sparsification} of $M$ if it is obtained by setting
to $0$ some of the entries of $M$ and (eventually) rescaling the non-zero
entries. 
A sparsification is useful if most of its entries are zero,
and yet 
the resulting matrix has a leading eigenvector close to the original one.
Given a sparsified matrix $S$, power method
can be applied by $x^{(t+1)}_i = \sum_{j=1}^nS_{ij}x^{(t)}_j$. 
If $S$ has $d$-nonzero entries per row, each node needs to 
communicate $d$ real numbers, and to perform 
$d$ sums and multiplications. For wireless devices, the badwidth
scales at most like $nd$. 

In \cite{AM02} Achlioptas and McSherry showed that a sparsification
can be constructed such that 
\begin{eqnarray}
\|M-S\|_2\le \theta \, \|M\|_2\, ,\label{eq:NormBound}
\end{eqnarray}
with only $d= O(1/\theta^2)$ non-zero entries per row. 
The inequality
(\ref{eq:NormBound}) immediately implies that computing the leading 
eigenvector of $S$, yields an estimator $\hu$ that satisfies 
$\|\hu-u\| \le 2\theta$. 
(Here and below, for $v,w\in\reals^m$, $v^\ast$ denotes its transpose 
and $\<v,w\>=v^\ast w$ denotes the scalar product of two vectors.
Let $\|v\|=\<v,v\>$ denote its Euclidean --or $\ell_2$-- norm, 
i.e. $\|v\|^2\equiv\sum_{i=1}^nv_i^2$. 
For a matrix $A$, $\|A\|_2$ denotes its $\ell_2$ operator norm,
i.e. $\|A\|_2\equiv \sup_{v\neq 0}\|Av\|/\|v\|$.)
The construction of \cite{AM02} is based on random sampling. Each entry 
of $M$ is set to $0$ independently with a given probability $1-p=1-d/n$.
Non-zero entries are then rescaled by a factor $1/p$. The 
bound (\ref{eq:NormBound}) is proved to hold with high probability 
with respect to the randomness in the sparsification.% procedure. 

While this approach is simple and effective, it still presents important 
shortcomings: $(i)$ For a fixed per node complexity 
which scales like $1/\theta^2$, 
this procedure achieves precision $\theta$: 
can one achieve a better scaling?
$(ii)$ A fixed subnetwork $G$ of the complete graph (corresponding to the 
sparsity pattern of $S$) needs to be maintained through the whole process.
This can be challenging in the presence of fading or of node 
failures/departures. $(iii)$ The target precision is to be decided 
at the beginning of the process, when the sparsification is constructed.

In this paper we use sparsification as a primitive and propose 
a new way to exploit its advantages. Roughly speaking at each round 
$t$ a new independent sparsification $S^{(t)}$ of $M$ is produced.
Estimates of the leading eigenvector are generated by applying 
$S^{(t)}$, i.e. through 
\begin{eqnarray}
x^{(t)}=S^{(t)}x^{(t-1)}\, ,\label{eq:Basic}
\end{eqnarray}
and then averaging across iterations $\hu^{(t)} \propto 
\sum_{\ell\le t}x^{(\ell)}/\|x^{(\ell)}\|$. 
We will refer to this algorithm as \gossip_pca.
In the limit case in which $S^{(t)}$ are in fact deterministic
and coincide with a fixed $S$, the present scheme reduces to the previous one.
However, general independent random sparsifications $S^{(t)}$ can 
model the effect of fading, short term link failures, node departures. 
(While complete independence is a simplistic model for these effects,
it should be possible to include short time-scale correlations
in our treatment.) 
Finally, the use of truly random, independent sparsifications
might be a choice of the algorithm designer.

Does the time-variability of $S^{(t)}$ deteriorate the algorithm precision?
Surprisingly, the opposite turns out to be true:
Using independent sparsifications appears to benefit accuracy 
by effectively averaging over a larger sample of the entries of $M$. 
As an example consider the sparsification scheme mentioned
above, namely each entry of $S^{(t)}$ is set to $0$
independently with a fixed probability $1-p$.
%(this can model independent link failures). 
Then, with respect to the total per-node computation and communication budget,  
scaling of the $\ell_2$ error $\|\hu-u\|$  
remains roughly the same as in the time-independent case (see Section \ref{sec:example}). 
Remarkably, the way optimal accuracy is achieved is significantly 
different from the one that is optimal within the time-independent case.
In the latter case it is optimal to invest resources in the densest
possible sparsification $S$, and then iterate it a few times. 
Within the present approach, one should rather use much sparser matrices
$S^{(t)}$ and iterate the basic update (\ref{eq:Basic}) many more times.
The use of sparser subnetworks %at each iteration 
is advantageous
both for robustness and the overhead of maintaining/synchronizing 
such networks.

Our main analytical result is
an error bound for the time-dependent iteration (\ref{eq:Basic}), 
that takes the form
\begin{eqnarray}
\|\hu^{(t)}-u\| \le C\left({\theta}/{\sqrt{t}}\,+\,\theta^2\log(1/\theta)^2\right)\, ,
\label{eq:UBInformal}
\end{eqnarray}
with a constant $C$ explicitly given below. 
Notice that,  
for $t$ large enough, this yields an error
roughly of size $\theta^2$. 
While using the same number of communications per node, 
this is significantly smaller 
than the error $\theta$ obtained by computing the leading eigenvector 
of a single sparsification. 
%However, as mentioned previously, 
%when we compare the achievable accuracy 
%with respect to the total amount of communication and 
%allow two schemes to use different per round communications, 
%the scaling of the error remains roughly the same 
%as when single sparsificatin is used.
%Detailed comparison is provided in Section \ref{sec:example}. 

The upper bound (\ref{eq:UBInformal}) holds under the following 
three assumptions:
$(i)$ $\|M-S^{(\ell)}\|_2\le\theta\|M\|_2$ for all $\ell\le t$;
$(ii)$ $\E(S^{(\ell)}) = M$; $(iii)$ $S^{(\ell)}$ invertible
for all $\ell$. Further it is required that the initial condition satisfies 
$\|x^{(0)}-u\| \le C\theta$. This can be generated by iterating
a fixed sparsification (say $S^{(1)}$) for a modest number of iterations 
(roughly $\log(1/\theta)$). Numerical simulations and heuristic arguments
further suggest that the last assumption is actually a proof artifact and
not needed in practice (see further discussion in Section
\ref{sec:analysis}).

The rest of the paper is organized as follows: 
Section \ref{sec:Main} provides a formal description of our 
algorithms and of our general performance guarantees. In Section
\ref{sec:example} we discuss implications of our analysis in specific 
settings. 
Section \ref{sec:related} reviews related work on randomized low complexity methods. 
Section \ref{sec:proof2} describes the proof of our main 
theorem. This leverages on the theory of products of random matrices,
a line of research initiated by Furstenberg and Kesten in the sixties
\cite{Furstenberg}, with remarkable applications in dynamical systems theory 
\cite{Oseledets}.
The classical theory focuses however on matrices of fixed dimension, 
in the limit of an infinite number of iterations, while here we 
are interested in high-dimensional (large $n$) applications.
We need therefore to characterize the tradeoff between dimensions 
and number of iterations.
In Section \ref{sec:technical}, we provide the proof of the technical lemmas used in the main proof. 
Finally, Section \ref{sec:value} 
discusses extending our algorithm to estimate the largest eigenvalues and 
provides a general performance guarantee.

%
%=========================================================================
%
\section{Main results} 
\label{sec:Main}

In this section, we spell out the algorithm 
execution and state the main performance guarantee.

\subsection{Algorithm}
\label{sec:algorithm}

As mentioned in the previous section $M\in\reals^{n\times n}$ is a symmetric
matrix, with eigenvalues $\lambda_1,\lambda_2,\dots,\lambda_n$.
Without loss of generality, we assume that the largest eigenvalue 
$\lambda_1$ is positive.
Further, we assume $\lambda_1>|\lambda_2|$ strictly. 
We will also write $\lambda \equiv\lambda_1$ and $u$ for the corresponding
eigenvector.
We assume to have at our disposal a primitive that
outputs a random sparsification $\Ms$ of  $M$.
A sequence of independent such sparsifications
will be denoted by $\{\Mt[1],\Mt[2],\dots\}$.
In the next two paragraphs we describe a centralized version of the algorithm,
and then the fully decentralized one.

\subsubsection{Centralized algorithm}

The system is initialized to a vector $x^{(0)}\in\reals^n$. 
Then we iteratively multiply the i.i.d. sparsifications 
$\Mt[1], \Mt[2],\dots$ to get a sequence of 
vectors $x^{(1)},x^{(2)},\dots$.
After $t$ iterations, our estimate for the leading eigenvector $u$ is 
\begin{eqnarray}
\hu^{(t)} = c(t)\sum_{s=1}^t\frac{x^{(s)}}{\|x^{(s)}\|}\, ,\label{eq:IterAverage}
\end{eqnarray}
with $c(t)$ the appropriate normalization to ensure $\|\hu^{(t)}\|=1$.

Note that, even after normalization, there is a residual sign ambiguity:
both $u$ and $-u$ are eigenvector. When in the following 
we write that $\hu^{(t)}$ approximate $u$ within a certain accuracy, it is
understood that $\hu^{(t)}$ does in fact approximate the closest of
$u$ and $-u$. A more formal resolution of this ambiguity 
uses the projective manifold define in Section \ref{sec:proof2}.

\subsubsection{Decentralized algorithm}

The algorithm described so far uses the following operations:
\noindent$(i)$ \emph{Multiplying vector $x^{(t-1)}$ by $S^{(t)}$,}
cf. Eq.~(\ref{eq:Basic}). If $S^{(t)}$
has $dn$ non-zero elements, this requires $O(d)$ operations per node per
round.\\
\noindent$(ii)$ \emph{Computing the normalizations $\|x^{(1)}\|$,
$\|x^{(2)}\|$,\dots,$\|x^{(t)}\|$.} Since $\|x^{(\ell)}\|^2 = \sum_{i=1}^n
(x^{(\ell)}_i)^2$, this task can be performed via a standard 
gossip algorithm. This entails an overhead of $\log(1/\ve)$ per node per 
iteration for a target precision $\ve$. We will neglect this 
contribution in what follows.\\
\noindent$(iii)$ \emph{Averaging normalized vectors across iterations,}
cf. Eq.~(\ref{eq:IterAverage}). 
Since node $i$ keeps the sequence of estimates $x^{(1)}_i$,\dots,
$x^{(t)}_i$,
this can be done without communication overhead, with $O(1)$ computation per 
node per iteration.

\vspace{0.1cm}

Finally the normalization constant $c(t)$ in Eq.~(\ref{eq:IterAverage})
needs to be computed. 
This amounts to computing the norm of the vector
on the right hand side of Eq.~(\ref{eq:IterAverage}), which 
is the same operation as in step $(2)$ (but has to be carried out only once).
From this description, it is clear that operation $(1)$ (matrix-vector 
multiplication) dominates the complexity and we will focus on this
in our discussion below and in Section \ref{sec:example}.
%
%=========================================================================
%
\subsection{Analysis}
\label{sec:analysis}

The algorithm design/optimization amounts to the choice of 
number of iterations $t$ and the the sparsification
method, which produces the i.i.d. matrices $\{S^{(\ell)}\}$.
The latter is characterized by two parameters:
$\theta$ which bounds the sparsification 
accuracy as per Eq.~(\ref{eq:NormBound}), and $d$,
the average number of non-zero entries per row, which determines 
its complexity.

The trade-off between $d$ and $\theta$ depends on the sparsification 
method and will be further discussed in the next section.
Our main result bounds the error of the algorithm in terms
of $\theta$, $t$ and of a characteristic of the matrix $M$, namely 
the ratio of the two largest eigenvalues $l_2= |\lambda_2|/\lambda$.
The proof of this theorem is presented in Section \ref{sec:proof2}.

\begin{thm}\label{thm:eigenvector}
Let $\{\Mt[\ell]\}_{\ell\ge 1}$ be a sequence of i.i.d. 
$n\times n$ random matrices such that 
$\E[\Mt[\ell]]=M$, $\|\Mt[\ell]-M\|_2\leq\theta\|M\|_2$,
$\Mt[\ell]$ is almost surely non-singular, and there is no proper
subspace $V\subseteq \reals^n$ such that $\Mt[\ell]V\subseteq V$
almost surely. 
Further, let $x^{(0)}\in\reals^n$
be such that $\|x^{(0)}-u\| \le \theta/(1-l_2)$ for the leading
eigenvector of $u$. Let the eigenvector estimates be
defined as per Eq.~(\ref{eq:Basic}) and (\ref{eq:IterAverage}).
Finally assume $\theta\leq(1/40)(1-l_2)^{3/2}$ and let $l_2\equiv|\lambda_2|/\lambda$.

Then, with probability larger than $1-\max(\delta,16/n^2)$, 
\begin{eqnarray}
\|\hu^{(t)}-u\| \leq  \frac{18\theta}{(1-l_2)\sqrt{t\delta}} +  
	 12\Big(\frac{\theta\log(1/\theta)}{(1-l_2)}\Big)^2 \;.
\label{eq:MainBound}
\end{eqnarray}
\end{thm}

The assumption on the samples $\{S^{(\ell)}\}_{\ell\ge 0}$ are rather
mild. The matrix whose eigenvector we are computing is the expectation
of $S^{(\ell)}$, the variability of $S^{(\ell)}$ is bounded in
operator norm, and finally the $S^{(\ell)}$ are sufficiently random
(in particular the do not share an eigenvector \emph{exactly}). The
latter can be ensured by adding arbitrarily small random perturbation to $S^{(\ell)}$.

At first sight, the assumption $\|x^{(0)}-u\| \le \theta/(1-l_2)$ 
on the initial condition might appear unrealistic:
the algorithm requires as input an approximation of the 
eigenvector $u$. A few remarks are in order. 
\emph{First,} the accuracy of the output, see Eq.~(\ref{eq:MainBound}),
is dramatically higher than on the input for $t=\Omega(1/\theta^2)$. 
In the following section, we will see that this is indeed the correct scaling of $t$ 
that achieves optimal perormance.
\emph{Second,} numerical
simulations show clearly that, for $\bx^{(t)}=x^{(t)}/\|x^{(t)}\|$,
the condition
$\|\bx^{(t)}-u\| \le \theta/(1-l_2)$ is indeed satisfied after a few
iterations.
 The heuristic argument is that the leading 
eigenvectors of $S^{(1)}$, $S^{(2)}$, \dots $S^{(t)}$ are roughly
aligned with $u$, and their second eigenvalues are significantly
smaller. Hence the scalar product $Z_t\equiv\<u,\bx^{(t)}\>$ behaves
approximately as a random walk with drift pushing out of $Z_t=0$.
Even if $Z_0=0$, random fluctuations produce a non-vanishing $Z_t$,
and the drift amplify this fluctuation exponentially fast.
The arguments in Section \ref{sec:proof2} further confirm 
this heuristic argument. For instance we will prove that 
the set $\|\bx^{(t)}-u\| \le \theta/(1-l_2)$ is absorbing,
in the sense that starting from such a set, the power iteration
keeps $\bx^{(t)}$ in the same set. On the other hand, starting from
any other point, there is positive probability of reaching the
absorbing set.
Finally, further evidence is provided by the fact that random initialization is sufficient 
for the eigenvalue estimation as proved in Section \ref{sec:value}.

As an example, we randomly generated a marix $M$ 
%of dimensions $1000\times1000$, with $l_2=0.9$. We
and computed $\bx^{(t)}=x^{(t)}/\|x^{(t)}\|$ according to \eqref{eq:Basic} 
using random sparsifications with $dn$ entries. 
%where each entry is sampled with probability $d/1000$ for $d\in\{40,80,160,320\}$.
Let $\tau=\arg\min_t \{\|\bx_{\rm rand}^{(t)}-\bx_{\rm u}^{(t)}\|\leq 0.001\}$. 
The subscript denotes two different initializations: %, but keeping the same sparsifications:  
$x^{(0)}_{\rm rand}$ is initialized with i.i.d Gaussian entries, 
and $x^{(0)}_{u}=u$. 
The following result illustrates that after a few iterations $t=O(\log(1/\theta))$, 
$\bx^{(t)}$ achieves error of order $\theta$ with $d=O(1/\theta^2)$ operations per node per round.
\begin{center} 
\vspace{-0.4cm}
\begin{tabular}{ | c || c | c | c | c | } 
\hline $d$        & 40 & 80 & 160 & 320 \\ 
\hline $\tau$    & 5.1 & 4.8 & 4.2 & 3.7 \\ 
\hline $\|\bx_{\rm rand}^{(\tau)}-u\|$ & 0.1110 & 0.0761 & 0.0521 & 0.0329 \\ 
\hline 
\end{tabular} 
\end{center}
\vspace{-0.3cm}
Finally, constructing a rough approximation of the leading 
eigenvector is in fact an easy task by multiplying the same 
sparsification $\Mt$ a few times. This claim is made precise by the 
following elementary remark.
\vspace{-0.2cm}
\begin{remark}
\label{rem:drift}
Assume that $x^{(0)}$ have i.i.d. components $\normal(0,1/n)$,
and define $x^{(t)} = \Mt[t]x^{(t-1)}$ where for 
$t\le t_*$, $\Mt[t] =\Ms$ is time independent and satisfies
$\|\Ms-M\|_2\le (\theta^2/2(1-l_2))\|M\|_2$.
If $t_*\ge 3\log(n/\theta)/(1-l_2-\theta)$, then 
$\|\bx^{(t_*)}-u\| \le \theta/(1-l_2)$ with probability at least
$1-1/n^2$.
\end{remark}
\vspace{-0.2cm}
The content of this remark is fairly intuitive:
the principal eigenvector of $\Ms$ is close to $u$, and
the component of $x^{(t)}$ along it grows exponentially faster than
the other components. A logarithmic number of iterations is then sufficient
to achieve the desired distance from $u$.

Finally, consider the assumption $\E[\Mt[\ell]]=M$.
In practice, it might be difficult to produce unbiased sparsifications:
does Theorem \ref{thm:eigenvector} provide any guarantee in this case?
The answer is clearly affirmative. Let $\E[\Mt[\ell]]=M'$ and
assume $\|M-M'\|_2\le \theta'\|M\|_2$. Then, it follows immediately from 
\eqref{eq:MainBound} that
\begin{eqnarray*}
\|\hu^{(t)}-u\| \leq \frac{18\theta}{(1-l_2)\sqrt{t\delta}} +  
	 12\Big(\frac{\theta\log(1/\theta)}{(1-l_2)}\Big)^2 + \frac{2\theta'}{1-l_2}\;,
\end{eqnarray*}
In other words the eigenvector approximation degrades gracefully with the 
quality of the sparsification.

%
%=========================================================================
%
\section{Examples and applications}
\label{sec:example}

In this section we apply our main theorem to specific settings
and point out possible extensions.

\subsection{Computation-accuracy tradeoff}

As mentioned above, Theorem \ref{thm:eigenvector}
characterizes the scaling of accuracy with the quality
of the sparsification procedure. For the sake 
of simplicity, we will consider the case in 
each entry of $M$ is set to $0$ independently with a fixed 
probability $1-d/n$, and non-zero entries are rescaled. 
In other words $S_{ij}=(n/d)M_{ij}$ with probability $(d/n)$,
and $S_{ij}=0$ otherwise.
This scheme was first analyzed in \cite{AM02}, but  
the estimate only holds for $d\ge (8\log n)^4$.
This condition was refined in \cite{KMO09}. Noting 
that for $d>\log n$ the maximum number of entries per row is of order
$d$, the latter gives
$$\|M-S\|_2\leq ({C}/{\sqrt{d}})\|M\|_2 \equiv \theta \|M\|_2\;.$$
In other words i.i.d. sparsification of the entries
yields $\theta = O(1/\sqrt{d})$. 
Further, denoting the total
complexity per node by $\chi$, we have $\chi \sim td$
either in terms of communication or of computation. 

In order to compute a computation-accuracy tradeoff
we need to link the accuracy to $t$ and $\theta$. 
Let us first consider the case in which a single sparsification
$S$ is used by letting $x^{(t)}= Sx^{(t-1)}$ and $\hu^{(t)} = 
x^{(t)}/\|x^{(t)}\|$. This procedure converges exponentially fast
to the leading eigenvector of $S$ which in turn 
satisfies $\|\hu^{(\infty)}-u\| \le 2\theta\le C'(1/\sqrt{d})$. 
Therefore if we denote by $\Delta_{\rm PM}\equiv \|\hu^{(t)}-u\|$ the 
corresponding error after $t$ iterations, we have 
\begin{align*}
\Delta_{\rm PM} \sim \theta + e^{-at}\, ,
\end{align*}
where we deliberately omit constants since we are only interested in 
capturing the scaling behavior.

Now we assume that we have a limit on the total complexity $\chi\sim td$,
and minimize the error $\Delta_{\rm PM}$
under this resources constraint, using the relation $\theta\sim 1/\sqrt{d}$. 
A simple calculation shows that the smallest error 
is achieved  when $t=\Theta(\log \chi)$  yielding 
\begin{align}
 \Delta_{\rm PM} \sim \sqrt{{(\log \chi)}/{\chi}} \;.
	\label{eq:timeindepend}
\end{align}
Next consider the algorithm developed in the present paper, \gossip_pca. 
The only element to be changed in our analysis is the relation
between accuracy and the parameters $\theta$ and $t$.  
From Theorem \ref{thm:eigenvector} we know that our estimator achieves 
error $\Delta_{\rm Gossip} = \|\hu^{(t)}-u\|$ that scales as
\begin{align*}
\Delta_{\rm Gossip} \sim {\theta}/{\sqrt{t}} \,+\, \big(\theta\log(1/\theta)\big)^2 \;,
\end{align*}
where again we omit constants. 
It is straightforward to minimize 
this expression under the constraints $\chi\sim td$,
and  $\theta\sim 1/\sqrt{d}$. 
The best scaling is achieved 
when $t=\Theta( \sqrt{\chi}/(\log\chi)^2)$ and 
$\theta=\Theta(1/(\chi^{1/4}\log\chi))$ yielding 
\begin{align}
\Delta_{\rm Gossip}\sim {1}/{\sqrt{\chi}} \;.
	\label{eq:timedepend}
\end{align}

Comparing \eqref{eq:timeindepend} and \eqref{eq:timedepend}, 
the scaling of the error with the per-node computation and communication 
remains roughly the same up to a logarithmic factor. 
Surprisingly, the way the best accuracy is achieved is significantly different. 
In the time-independent case (the standard power method), 
it is optimal to invest a lot of resources in one iteration
with a dense matrix $S$ that has $d=\Theta(\chi/\log \chi)$ non-zero 
entries per row.
In return, only a few iterations $t=\Theta(\log\chi)$ are required. 
Within the proposed time-dependent gossip 
approach, one should rather 
use a much sparser matrices $S^{(\ell)}$ with 
$d=\Theta(\sqrt{\chi}(\log \chi)^2)$ non-zero entries per row 
and use a larger number of iterations $t=\Theta(\sqrt{\chi}/(\log\chi)^2)$.

%These computation-accuracy tradeoffs hold up to numerical constants. 
To illustrate how the two gossip algorithms compare in practice, 
we present results of a numerical experiment 
from the positioning application. 
From $1000$ nodes placed in the $2$-dimensional unit square uniformly at random, 
we define the matrix of squared distances.  
Let $p_i$ be the position of node $i$, then $D_{ij}=\|p_i-p_j\|^2$. 
After a simple centering operation, 
the top two eigenvectors reveal 
the position of the nodes up to a rigid motion (translation and/or rotation) \cite{OKM10}. 
We can extend the gossip algorithms to estimate the first two eigenvectors 
as explained in Section \ref{sec:Extensions}. 
Let the columns of $U\in\reals^{1000\times 2}$ be the first two eigenvectors and 
$\|\cdot\|_F$ be the Frobenius norm of a matrix such that $\|A\|_F^2=\sum_{i,j}A_{ij}^2$.
Denote by $\Delta(d)=(1/\sqrt{2})\|U-\hU\|_F$ the resulting error for a particular choice of $d$. 

To simulate a simple gossip setting with constrained communication, 
we allow $d$ to be either $50$ or $500$. 
For the two gossip algorithms and  
for each value of the total complexity $\chi$, 
we plot the minimum error achieved 
using one of the two allowed communication schemes: $\min_{d\in\{50,500\}}\Delta(d)$.
For comparison, performance of the power method  
on complete dense matrices is also shown (see Section \ref{sec:intro}).
As expected from the analysis, 
\gossip_pca achieved smaller error with sparse matrices ($d=50$) for all values of $\chi$. 
When a single sparsification is used, 
there is a threshhold at $\chi=14500$, above which 
a dense matrix ($d=500$) achieved smaller error. 
Notice a discontinuity of the derivative at the threshold. 
%Another implication of this example is that the condition $\|\bx^{(t)}-u\| \le C\theta$ is indeed satisfied after a few iterations.

\begin{figure}
 \centering
 \includegraphics[width=6cm]{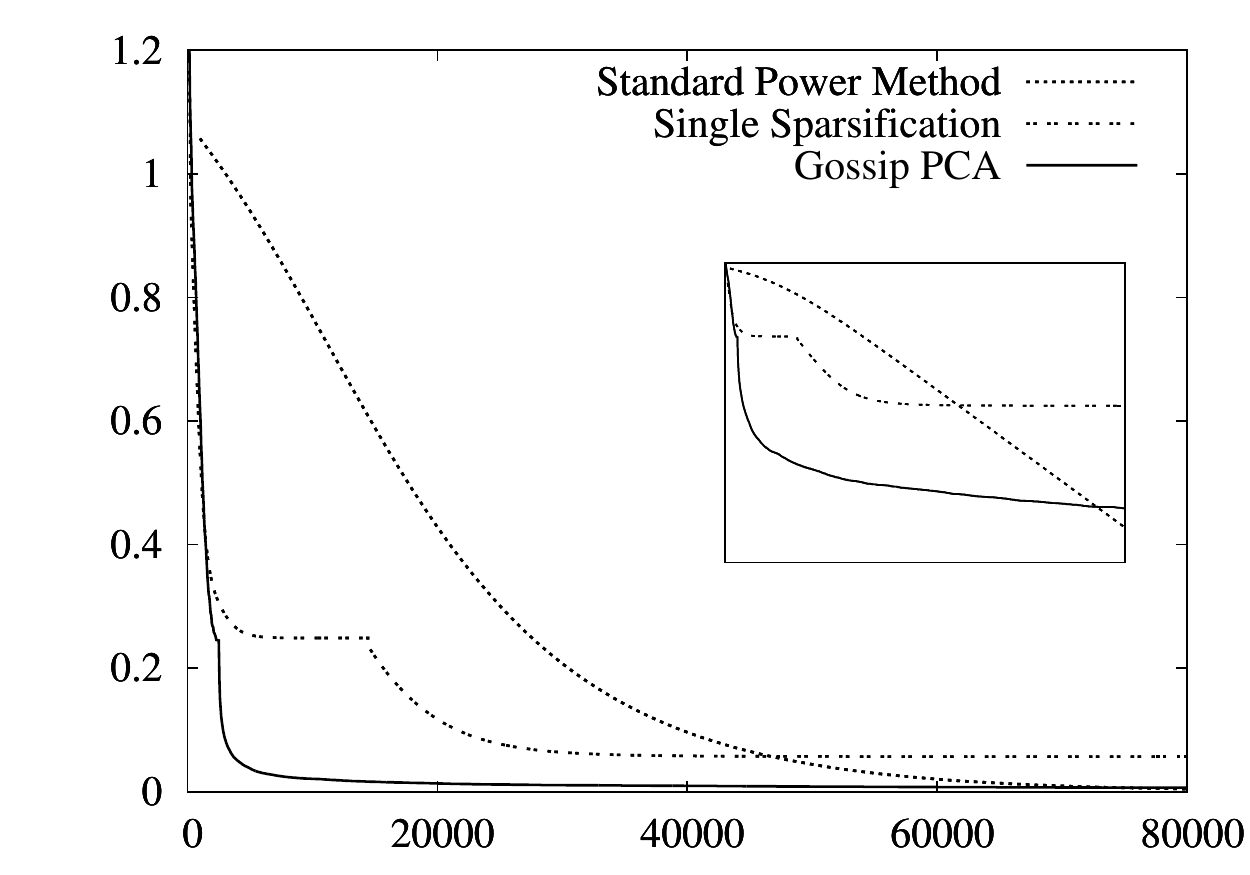}
 \put(-167,34){\begin{sideways}\small Estimation error\end{sideways}}
 \put(-130,-5){\small Total complexity per node}
 \caption{Eigenvector estimation error against complexity. In the inset the result is plotted in log-scale.}
 % position.png: 720x796 pixel, 85dpi, 21.52x23.79 cm, bb=0 0 610 674
 \label{fig:positioning}
  \vspace{-0.5cm}
\end{figure}

%
%****************************************************
%
\subsection{Comparison with gossip averaging}
\label{sec:Averaging}

Gossip methods have been quite successful in computing symmetric
functions of data $\{x^{(0)}_i\}_{1\le i\le n}$ available at the nodes.
The basic primitive in this setting
is a procedure computing the average
$\sum_{i=1}^nx^{(0)}_i/n$. This algorithm shares similarities with 
the present one. One recursively applies independent random 
matrices $P^{(1)}$, $P^{(2)}$, \dots according to:
\begin{eqnarray}
x^{(t)}=P^{(t)} x^{(t-1)}\, ,\label{eq:Averaging}
\end{eqnarray}
where $P^{(t)}$ is the matrix that averages entries $i(t)$ and $j(t)$
of $x^{(t)}$ (in other words it is the identity outside a $2\times 2$
block corresponding to coordinates $i(t)$ and $j(t)$). 

It is instructive to compare the two problems. 
In the case of simple averaging, one is interested in approximating
the action of a projector $P$, namely the matrix with 
all entries equal to $1/n$. In eigenvector calculations
the situation is not as simple, because the matrix of interest $M$
is not a simple projector. In both cases we approximate this action
by products of i.i.d. random matrices whose 
expectation matches the matrix of interest. However in averaging,
the leading eigenvector of $P$ is known \emph{a priori},
it is the constant vector $u = (1/\sqrt{n},\dots,1/\sqrt{n})$.
As a consequence, sparsifications  $P^{(t)}$ can be constructed in such a way
that $P^{(t)} u =u$ with probability $1$.

Reflecting these differences, the behavior of the present algorithm 
is qualitatively different from gossip averaging. Within the 
latter $x^{(t)}$ converges asymptotically to the constant vector, whose
entries are equal to $\sum_{i=1}^nx^{(0)}_i/n$. The convergence rate depends
on the distribution of the sparsification $P^{(t)}$. In 
\gossip_pca, the sequence of normalized vectors $x^{(t)}/\|x^{(t)}\|$
does not converges to a fixed point. The distribution
of  $x^{(t)}/\|x^{(t)}\|$  instead converges to
a non-trivial stationary distribution whose mean is approximated by 
$\hu^{(t)}$. An important step in the proof of Theorem
\ref{thm:eigenvector} consists in showing that the mean of this 
distribution is much closer to the eigenvector than a typical 
vector drawn from it.

%
%***********************************************************
%
\subsection{Extensions}
\label{sec:Extensions}

It is worth pointing out some extensions of our results, and 
interesting research directions:

{\em More than one eigenvector.} In many applications of interest, 
we need to compute $r$ leading eigenvectors, where $r$ is larger than one, 
but typically a small number. 
In the case of positioning wireless devices, 
$r$ is consistent with the ambient dimensions, hence $r=3$.
As for the standard power iteration, the algorithm 
proposed here can be generalized to this problem.
At iteration $t$, the algorithm keeps track of 
$r$ orthonormal vectors $x^{(t)}(1)$, \dots $x^{(t)}(r)$. In the distributed 
version, node $i$ stores the $i$-th coordinate of each 
vector, thus requiring $O(r)$ storage capability.
The vectors are updated by letting $\widetilde{x}^{(t)}(a) = \Mt[t]x^{(t)}(a)$.
and then orthonormalizing $\widetilde{x}^{(t)}(1)$, 
\dots $\widetilde{x}^{(t)}(r)$
to get $x^{(t)}(1)$, \dots $x^{(t)}(r)$. Orthonormalization can be done
locally at each node if it has access to the Gram matrix
$G = (G_{ab})_{1\le a,b\le r}$
\begin{eqnarray}
G_{ab} \equiv \frac{1}{n}\sum_{i=1}^n \widetilde{x}_i^{(t)}(a)
\widetilde{x}_i^{(t)}(b)
\end{eqnarray}
This can be computed via gossip averaging, using messages consisting of 
$r(r+1)/2$ real numbers. Therefore the total communication 
complexity per node per 
iteration is of order $r^2\log(1/\ve)$ to achieve precision $\ve$.
Indeed, such distributed orthonormalization procedure was studied 
in \cite{KM08} for decentralized implementation of the standard power method.

\emph{Richer stochastic models for random sparsification.}
Our main result holds under the assumption that $\Mt[1]$,
$\Mt[2]$,\dots,$\Mt[t]$ are i.i.d. sparsifications of the matrix $M$.
This is a reasonable assumption when the random sparsifications
are generated by the algorithm itself. The same assumption can also model
short time-scale link failures, as due for instance to fast fading
in a wireless setting. On the other hand, a more accurate
model of link failures would describe $\Mt[1]$,
$\Mt[2]$,\dots,$\Mt[t]$ as a stochastic process. We think that our
main result is generalizable to this setting under appropriate
ergodicity assumptions on this process. More explicitly,
as long as the underlying stochastic process mixes (i.e. loses
memory of its initial state) on time scales shorter than $t$,
the qualitative features of Theorem \ref{thm:eigenvector}
should remain unchanged. 
Partial support of this intuition is provided 
by the celebrated Oseledets' multiplicative ergodic theorem
that guarantees convergence the exponential growth rate of $\|x^{(t)}\|$
in a very general setting \cite{Oseledets} (namely within the context of 
ergodic dynamical systems).

\emph{Communication constraints: Rate and noise.} 
%Consider the case in which
%the matrix $M$ itself is sparse, or a fixed sparsification 
%$\Ms$ is used within the ordinary power method. 
%In a decentralized setting 
%it is unavoidable to take into consideration communication rate
%constraints, and communication errors. The presence of errors implies
%that the actual matrix $\Mt[t]$ used at iteration $t$ is not
%identical to $\Ms$, but is rather a perturbation of it.
%The effect of noise can then be studied through Theorem \ref{thm:eigenvector}.
%Rate constraints imply that real numbers cannot be 
%communicated through the network, unless some quantization is used.
%An approach consists in using some form of randomized rounding
%for quantization. In this case the matrices $\Mt[1]$,\dots,$\Mt[t]$
%can be interpreted as i.i.d. quantizations of a fixed matrix $S$.
%Theorem \ref{thm:eigenvector} implies that, roughly speaking,
%the error in the eigenvector computed with this approach 
%scales quadratically in the quantization step.
%(Notice that quantization also affects the vector on the right-hand side of
%Eq.~(\ref{eq:Basic}), but we expect this effect to be roughly of the same 
%order as the effect of the quantization of $\Mt[t]$.)
In a decentralized setting, 
it is unavoidable to take into consideration communication rate
constraints and communication errors. 
The presence of errors implies
that the actual matrix used at iteration $t$ is not
$\Mt[t]$ but is rather a perturbation of it.
The effect of noise can then be studied through Theorem \ref{thm:eigenvector}.
Rate constraints imply that real numbers cannot be 
communicated through the network, unless some quantization is used.
An approach consists in using some form of randomized rounding
for quantization. 
In this case, the effect of quantization can also be 
studied through Theorem \ref{thm:eigenvector}.
This implies that, roughly speaking,
the error in the eigenvector computed with this approach 
scales quadratically in the quantization step.
(Notice that quantization also affects the vector on the right-hand side of
Eq.~(\ref{eq:Basic}), but we expect this effect to be roughly of the same 
order as the effect of the quantization of $\Mt[t]$.)
Further, when 
the matrix $M$ itself is sparse, or a fixed sparsification 
$\Ms$ is used within the ordinary power method, 
Theorem \ref{thm:eigenvector} can be used to study 
the effect of noise and quantization.

%
%=========================================================================
%
\section{Related Work} 
\label{sec:related}

The need for spectral analysis of massive data sets 
has motivated a considerable effort towards the development 
of randomized low complexity methods. A short 
sample of the theoretical literature in this topic includes 
\cite{DFKVV99,AM02,DK03,FKV04,DM05,DKM06}. 
Two basic ideas are developed in this line of research: 
\emph{sparsify} of the original matrix $M$ to 
reduce the cost of matrix-vector multiplication; \emph{apply} 
the matrix $M$ to a random set of vectors in order to approximate
its range. 
Both of these approaches are developed in a centralized setting 
where a single dataset is sent to a central processor.  
While this allows for more advanced algorithms than power iteration, 
these algorithms might not be directly applicable in 
a decentralized setting considered in this paper, 
where each node has limited computation and comunication capability 
and the datasets are often extremely large such that the data has to be stored in a distributed manner. 
%and detailed comparisons to the line of centralized algorithms is outside the scope of this paper. 

Fast routines for low-rank approximation are useful in many
areas of optimization, scientific computing and simulations.
Hence similar ideas were developed in that literature: 
we refer to \cite{TroppReview} for references and an
overview of the topic.

Kempe and McSherry \cite{KM08} studied a decentralized power iteration
algorithm for spectral analysis. 
They considered matrices that are inherently sparse. 
Therefore, no sparsification is used and 
all the entries are exploited at every iteration. 
Hence, their algorithm eventually computes the optimal low-rank approximation exactly.
%While not explicitly using sparsification, 
%their algorithm implicitly leverages the sparsity of $M$
%in internet-motivated applications (for instance, computing authority weights). 
The same paper introduced the decentralized 
orthonormalization mentioned in Section \ref{sec:Extensions}.

The idea of using a sequence of distinct sparsifications
to improve the accuracy of power iteration was not studied 
in this context. Somewhat related is the basic idea in randomized
algorithms for gossip averaging \cite{BoydGossip}. 
As discussed in Section \ref{sec:Averaging}, these algorithms operate
by applying a sequence of i.i.d. random matrices to an initial vector
of data. The behavior and analysis is however considerably simplified by
the fact that these matrices share a common leading eigenvector, that is known 
\emph{a priori}, namely the eigenvector $u=(1/\sqrt{n},\dots,1/\sqrt{n})$.
Overviews of this literature is provided by \cite{ShahReview}
and \cite{DimakisReview}. 
Quantization is an important concern in the practical implementation
of gossip algorithms, and has been studied in particular in the context
of consensus \cite{Quantized1,Quantized2}. As discussed in the 
last section, the effect 
of randomized quantization can also be included in the present setting.

Finally, there has been recent progress in the development
of  sparsification schemes that imply better
error guarantees than in Eq.~(\ref{eq:NormBound}), see for instance 
\cite{SpielmanSparse}. It would be interesting to study the effect of 
such sparsification methods in the present setting.

%
%=========================================================================
%

\section{Proof of the Main Theorem}
\label{sec:proof2}

In this section, we analyze the quality of the estimation
$\hu^{(t)}$ provided by our algorithm and prove 
Theorem \ref{thm:eigenvector}. Before diving into
the technical argument, it is worth motivating the main ideas. 
We are interested in analyzing the random trajectory
$\{x^{(t)}\}_{t\ge 0}$ defined as per Eq.~(\ref{eq:Basic}). 
One difficulty is that this process cannot 
be asymptotically stationary, since $x^{(t)}$ gets multiplied
by a random quantity. Hence it will either grow exponentially
fast or shrink exponentially fast.

A natural solution to this problem would be to track
the normalized vectors $\tx^{(t)}\equiv x^{(t)}/\|x^{(t)}\|$.
Also this approach presents some technical difficulty
that can be grasped by considering the special 
case in which $S^{(t)}=M$ for all $t$
(no sparsification is used). Neglecting exceptional initial conditions
(such that $\<x^{(0)},u\>=0$)
this sequence can either converge to $u$ or to $-u$. 
In particular, it cannot be uniformly convergent.
The right way to eliminate this ambiguity is to track the unit 
vectors $\tx^{(t)}$ `modulo overall sign'. The space of unit
vectors modulo a sign is the projective space $\P_n$,
that we will introduce more formally below. 

We are therefore naturally led to consider the random trajectory
$\{\bx_t\}_{t\ge 0}$ --indeed a Markov chain-- taking values
in the projective space $\bx_t\in \P_n$. 
We will prove that two important
facts hold under the assumptions of Theorem \ref{thm:eigenvector}:
$(1)$ The chain converges quickly to a stationary 
distribution $\mu$; $(2)$ The distance between the baricenter of 
$\mu$ and $u$ is of order $\theta^2$. 
Fact $(1)$ implies that $\hu^{(t)}$, cf. Eq.~(\ref{eq:IterAverage}), 
is a good approximation of the baricenter of $\mu$.
Fact $(2)$ then implies Theorem \ref{thm:eigenvector}.

In the next subsection we will first define formally the process
$\{\bx_t\}_{t\ge 0}$, and provide some background
(Section \ref{sec:MCGeneral}), and then 
present the formal proof (Section
\ref{sec:eigenvectorproof}), along the lines sketched above.
%
%**************************************************************
%
\subsection{The Kesten-Furstenberg Markov chain}
\label{sec:MCGeneral}

As anticipated above, we shall denote by 
$\P_n$ the projective space in $\reals^n$.
This is defined as the space of lines through the origin in $\reals^n$. 
Equivalently, $\P_n$ is  the space of equivalence classes 
in $\reals^n\setminus\{0\}$ for the equivalence relation  $\sim_{\P}$,
such that $x\sim_{\P} y$ if and only if 
$x=\lambda y$ for some $\lambda\in\reals\setminus\{0\}$.
This corresponds with the description given above,
since it coincides with  the space of equivalence classes 
in $S^n\equiv\{x\in\reals^n\,:\, \|x\|=1\}$ 
for the equivalence relation  $\sim_{\P}$,
such that $x\sim_{\P} y$ if and only if 
$x=\lambda y$ for some $\lambda\in\{+1,-1\}$.

In the future, we denote elements of $\P_n$ by boldface letters
$\bx,\by,\bz,\dots$ and the corresponding representatives
in $\reals^n$ by $x,y,z,\dots$. 
We generally take these representatives to have unit norm. 
We use a metric on this space defined as 
\begin{eqnarray*}
d(\bx,\by) \equiv \sqrt{1-\<x,y\>^2}\, .
\end{eqnarray*}
Random elements in $\P_n$ will be denoted by boldface capitals
$\bX,\bY,\bZ,\dots$.

An invertible matrix $\Ms\in\reals^{n\times n}$ acts naturally on $\P_n$,
by mapping $\bx\in\P_n$ (with representative $x$)
to the element $\by\in \P_n$ with representative of $Sx$ 
(namely the line through $Sx$, or the unit vector $Sx/\|Sx\|$
modulo sign). We will denote this action by writing
$\by = S\bx$, but emphasize that it is a non-linear map, since it implicitly
involves normalization.

Given a sequence of i.i.d. random matrices 
$\{\Mt[t]\}_{t\ge 1}$  that are almost surely invertible,
with common distribution $p_S$, 
we define the Markov chain $\{\bX_t\}_{t\ge 0}$ with values in $\P_n$ by 
letting
\begin{eqnarray}
\bX_t = \Mt[t]\Mt[t-1]\cdots \Mt[1]\bX_0 \;,\label{eq:Kesten}
\end{eqnarray}
for all $t\ge 1$.
We assume the following conditions:
\begin{enumerate}
\vspace{-0.4cm}
\item[\LA.] There exists no proper linear subspace $V\subseteq \reals^n$
such that $\Mt[1]V\subseteq V$ almost surely.
\vspace{-0.1cm}
\item[\LB.] There exist a sequence $\{\Mt[t]\}_{t\ge 1}$ in the support 
of $p_S$, such that letting $\Ms^T\equiv \Mt[T]\Mt[T-1]
\cdots\Mt[1]$, we have $\sigma_2(\Ms^T)/\sigma_1(\Ms^T)\to 0$
as $T\to\infty$.
\vspace{-0.4cm}
\end{enumerate}
It was proved in \cite{LePage}, 
that, under the assumptions $\LA$ and $\LB$, 
there exists a unique measure $\mu$ on $\P_n$ that is stationary for the Markov chain $\{\bX_t\}$. The Markov chain converges to the stationary
measure as $t\to\infty$ (we refer to the Appendix for a formal statement).

For the purpose of proving Theorem \ref{thm:eigenvector}, 
uniqueness of the stationary measure is not enough: we will need
to control the rate of convergence to stationarity.
We  present here a general theorem to bound the rate of convergence, 
and we will apply it to the chain of interest in the next section. 
Let us start by stating two more assumptions. We denote by
$\G \subseteq \P_n$ a  (measurable) subset of the projective space,
and assume that there exists a constant $\rho\in(0,1)$  
such that
\begin{enumerate}
\vspace{-0.4cm}
\item[\Aa.] For any $\bx \in \G$, $\Mt[t]\bx \in \G$ almost surely.
\vspace{-0.1cm}
\item[\Ab.] For any $\bx\neq\by\in \G$, $\E\left[d(\Mt[t]\bx,\Mt[t]\by)
\right] \le \rho \, d(\bx,\by)$.
\vspace{-0.4cm}
\end{enumerate}
We then have the following. 
\begin{theorem}\label{thm:muproperty}
Assume conditions $\LA$ and $\LB$ hold, together with 
$\Aa$ and $\Ab$.  Denote by $\mu$ the unique stationary measure 
of the Markov chain $\{\bX_t\}_{t\ge 0}$.
Then 
\begin{eqnarray*}
	\mu(\G^c) = 0\;.
\end{eqnarray*}
Further, if $\bX_{0}\in \G$ then for any $L$-Lipschitz 
function\footnote{We say that $f$ is $L$-Lipschitz if,
for any $\bx,\by\in\P_n$, $|f(\bx)-f(\by)|\le L\,d(\bx,\by)$.} 
$f:\P_n\to\reals$, we have
\begin{eqnarray*}
	\big|\E[f(\bX_t)]-\stat(f)\big|\le \, L\,\rho^t \;.
\end{eqnarray*}
\end{theorem}
The proof of this Theorem uses a coupling technique analogous 
to the one of \cite{LePage}. We present it in the appendix for greater 
convenience of the reader.
%
%**************************************************************
%
\subsection{Proof of Theorem \ref{thm:eigenvector}}
\label{sec:eigenvectorproof}

In this section we analyze the \gossip_pca algorithm using the 
general methodology developed above.
In particular, we consider the Markov chain (\ref{eq:Kesten})
whereby $\{S^{(\ell)}\}_{1\leq \ell\leq t}$ are i.i.d. sparsifications 
of $M$ satisfying the  conditions: 
$(i)$ $\|S^{(\ell)}-M\|_2\le \theta \|M\|_2$;
$(ii)$ $\E[S^{(\ell)}]=M$;
$(iii)$ $S^{(\ell)}$ is almost surely non-singular.
Throughout the proof, 
we let $\bu\in\P_n$ denote an element of $\P_n$ represented by $u$.

Note that the conditions $\LA$ stated in the previous section
holds by assumption in  Theorem \ref{thm:eigenvector}.
Further let $\lambda_1(\ell)$ and $\lambda_2(\ell)$ the largest
and second largest singular values of $\Mt[\ell]$. 
By assumption $(i)$, and since by hypothesis
$\theta\leq(1/40)(1-l_2)^{3/2}$, implying $\|\Mt[\ell]-M\|_2\le
(\lambda-|\lambda_2|)/2$, we have  
$|\lambda_1(\ell)/\lambda_2(\ell)|$ $>1$ almost surely.
Hence by taking $\Mt[1]=\Mt[2]=\dots=\Mt[T]=\dots$ in the support
of $p_S$, we have that condition $\LB$ holds as well. 

By applying the main theorem in \cite{LePage}
(restated in the Appendix), we conclude
that there exists a unique stationary distribution $\mu$
for the Markov chain $\{\bX_t\}$, and that the chain converges to it.

We next want to apply Theorem \ref{thm:muproperty} to bound the 
support and the rate of 
convergence to this stationary distribution.  
We define the `good' subset $\G\subseteq\P_n$ by
\begin{align}
\G = \left\{\bx\in\P_n : d(\bx,\bu) \leq \frac{2\theta}{1-l_2} \right\}\;. \label{eq:defG}
%	\G \equiv \left\{\bx\in\P_n : d(\bx,\bu) \leq \frac{\sqrt{(1-l_2)}}{20} \right\}\;. \label{eq:defG}
%
\end{align}
Our next lemma shows 
 assumptions $\Aa$ and $\Ab$ are satisfied in this set $\G$, 
with a very explicit expression for the contraction coefficient $\rho$.
\begin{lemma}
\label{lem:goodset}
	Under the hypothesis of Theorem \ref{thm:eigenvector}, 
	for any $\bx\in\G$ we have $S^{(\ell)}\bx\in\G$. 
	Further, for any $\bx\neq\by\in\G$, 
	letting 
	$\rho \equiv 1-(4/5)(1-l_2)\in (0,1)$, we have
	\begin{align*}
		\E d(S^{(\ell)}\bx,S^{(\ell)}\by) \leq \rho\,d(\bx,\by) \;.
	\end{align*}
\end{lemma}
The proof of this lemma can be found in the next section. 
As a consequence of this lemma we can apply Theorem \ref{thm:muproperty}.
In particular, we conclude  that $\mu$ is supported on the 
good set $\G$.

Next consider the estimate $\hu^{(t)}\in\reals^n$ produced 
by our algorithm, cf. Eq.~(\ref{eq:IterAverage}). 
 This is given in terms of the 
Markov chain on $\P_n$ by
\begin{align*}
\hu^{(t)} = \frac{\sum_{\ell=1}^t f(\bX_\ell)}{\|\sum_{\ell=1}^t f(\bX_\ell)\|} \;,
\end{align*} 
where we define $f:\P_n\mapsto \reals^n$ such that 
$f(\bx)$ 
is a representative of $\bx$ satisfying $\|f(\bx)\|=1$ and 
$\<u, f(\bx)\>\geq 0$. 
We use $\bU_t\in\P_n$ to denote an element in $\P_n$ represented by $\hu^{(t)}$. 

Let $\mu(f) = \int f(\bx)\mu(\de \bx)\in\reals^n$ be 
the expectation of $f(\,\cdot\,)$
with respect to the stationary distribution  
(informally, this is the baricenter of $\mu$).
With a slight abuse of notation, we
let $\mu(f)$ denote the corresponding element in $\P_n$ as well. 
Then, by the triangular inequality, we have, for any $t$, 
\begin{eqnarray*}
 	d(\bu,\bU_t) \leq d(\bu,\mu(f))+d(\bU_t,\mu(f)) \;.
\label{eq:Decomposition}  
\end{eqnarray*}
The left hand side is the error of our estimate of the leading eigenvector.
This is decomposed in two contributions: a deterministic one,
namely $d(\bu,\mu(f))$, 
that gives the distance between the leading eigenvector and 
the baricenter of $\mu$, and a random one i.e. 
$d(\bU_t,\mu(f))$, that measures the distance between the average of 
our sample and the average of the distribution.

In order to bound $d(\bU_t,\mu(f))$, 
we use the following fact that holds for any $a,b\in\reals^n$
\begin{eqnarray}
\sqrt{1-\frac{\<a, b\>^2}{\|a\|^2 \|b\|^2}} 
	\leq \frac{\|a-b\|}{\sqrt{\|a\|\|b\|}} \;.\label{eq:SimpleIneq}
\end{eqnarray}
This follows immediately from 
$2\|a\|\|b\|-2\<a, b\> \leq \|a-b\|^2$.
We apply this inequality to
$a=\mu(f)$ and $b=(1/t)\sum_{\ell=1}^t f(\bX_\ell)$. 
We need therefore to lower bound $\|\mu(f)\|$ and 
$\|(1/t)\sum_{\ell=1}^t f(\bX_\ell)\|$ and to upper bound 
$\|(1/t)\sum_{\ell=1}^t f(\bX_\ell)-\mu(f)\|$.

Denote by $\cP_u$ the orthogonal projector onto $u$.
From Theorem \ref{thm:muproperty}, we know that $\mu(\G^c)=0$. Hence, 
using $\theta<(1/40)(1-l_2)^{3/2}$, we have
$\|\mu(f)\| \geq \|\mu(\cP_u(f))\| \geq \sqrt{1-1/400}$. 
Similarly since $\bX_{0}\in\G$, we have by $\Aa$ that $\bX_{\ell}\in\G$ 
for all $\ell$, and therefore 
$\|(1/t)\sum_{\ell=1}^t f(\bX_\ell)\| \geq \sqrt{1-1/400}$. 
We are left with the task of bounding $\|a-b\|$.
This is done in the next  
lemma that uses in a crucial way Theorem \ref{thm:muproperty}.
\begin{lemma}
	\label{lem:NormConvergence}
	Under the hypothesis of Theorem \ref{thm:eigenvector}  
	\begin{align*}
	\bE\Big\|\frac1t\sum_{\ell=1}^t f(\bX_\ell) - \stat(f)\Big\|^2 \leq \frac{70\theta^2}{(1-l_2)^2t}\;.
	\end{align*}
\end{lemma}
Applying Markov's inequality and Eq.~(\ref{eq:SimpleIneq}), 
we get, with probability larger than $1- \delta/2$   
\begin{eqnarray*}
  	d(\bU_t,\mu(f)) \leq  \frac{12\theta}{(1-l_2)\sqrt{t\delta}}\;. 
\end{eqnarray*}

Next, we bound the term  $d(\bu,\mu(f))$ in 
Eq.~(\ref{eq:Decomposition})  with the following lemma. 
\begin{lemma}
	\label{lem:statdist}
	Under the hypothesis of Theorem \ref{thm:eigenvector},  
	\begin{align*}
	d(\bu,\mu(f)) \leq 8\Big(\frac{\theta\log(1/\theta)}{(1-l_2)}\Big)^2\;.
	\end{align*}
\end{lemma}
By noting that $\|u-\hu^{(t)}\|\leq \sqrt{2}\, d(\bu,\bhU_t)$, 
this finishes the proof of the theorem.

%
%*************************************************************************
%
\section{Proof of technical Lemmas}
\label{sec:technical}

%
%*************************************************************************
%
\subsection{Proof of Remark \ref{rem:drift}} 
\label{app:drift}

Assuming initial vector $X\in\reals^n$ with i.i.d. Gaussian entries, 
we can get close to $u$ by iteratively applying a single sparsification $S$. 
Define a good set of initial vectors 
\begin{eqnarray*}
 	\cF_n=\left\{x\in\reals^{n}: |u^\ast x|\geq\frac{1}{n^{5/2}}\text{ and }\max_{i\in[n]} |u_i^\ast x| \leq \sqrt{\frac{6\log n}{n}} \right\}\;.
\end{eqnarray*} 
Since, $u_i^\ast X$'s are independent and distributed as $\cN(0,1/n)$, 
it follows that we have $\prob(|u_i^\ast X|\geq\sqrt{(6\log n)/n})\leq2/n^3$ 
and $\prob(|u^\ast X|)\leq1/n^2$. 
Applying union bound, we get $\prob(X\in\cF_n)\geq1-3/n^2$. 
Assuming we start from this good set, 
we show that for $k$ large enough, 
we are guaranteed to have $\|u-\bx^{(k)}\|\leq\theta/(1-l_2)$.%\eqref{eq:lem_drift1}.
%$S^k\bx\in\G$, where $G$ is defined in \eqref{eq:defG}.
%\begin{eqnarray*}
%	d(\bu,Q^k\bx_0) \leq \frac{\sqrt{1-l_2}}{20} \;. 
%\end{eqnarray*}

Let $\{\tlam_i\}$ be the eigenvalues of $S$ 
such that $\tlam_1\geq|\tlam_2|\geq\cdots\geq|\tlam_n|$, and 
let $\{\tu_i\}$ be the corresponding eigenvectors. 
We know that $\tlam_1>0$ since $\tlam_1\geq\lambda-\|S-M\|_2$  
and $\|S-M\|_2<\lambda$ by assumption. 
Then, by the triangular inequality, 
\begin{eqnarray*}
 	\|u-\bx^{(k)}\| \leq \|u-\tu\| + \|\tu-\bx^{(k)}\| \;.
\end{eqnarray*}

%The first term is bounded by $d(\bu,\btu)\leq\theta\leq(\sqrt{1-l_2}/40)$, since $\|Q-M\|_2\leq\theta\lambda$. 
To bound the first term, note that 
\begin{eqnarray*}
\|M-S\|_2 &\geq& |u^t(M-S)u|  \\
	      &\geq& \lambda-\tlam_1(u^\ast\tu)^2-\tlam_2\|\cP_{\tu^{\perp}}(u)\|^2 \;.
\end{eqnarray*}
This implies that $(u^\ast\tu)^2\geq(\lambda-\tlam_2-\|M-S\|_2)/(\tlam_1-\tlam_2)$. 
We can further apply Weyl's inequality \cite{HJ90}, to get $|\tlam_i-\lambda_i |\leq\|M-S\|_2$. 
It follows that $(u^\ast\tu)^2\geq(\lambda-\lambda_2-2\|M-S\|_2)/(\lambda-\lambda_2+2\|M-S\|_2)$.
Note that this bound is non-trivial only if $\|M-S\|_2\leq(\lambda-\lambda_2)/2$. 
%We can easily find an example where $u\perp\tu$ when $\theta>(1-l_2)/2$. 
Using the fact that $(1-a)/(1+a)\leq(1-a)^2$ for any $|a|<1$, this implies that 
%For $\theta\leq(1/40)(1-l_2)$, this gives $(u^\ast\tu)^2\geq(19/20)^2(1-l_2)(1+(21/19)l_2)$.
%Using the fact that $\sqrt{1-a^2}\leq1-a/2$ for any $|a|<1$, we get 
\begin{eqnarray*}
 	%d(\bu,\btu) \leq 1-\frac{1}{2}\left(\frac{19}{20}\right)^2\left(1-l_2^2\right) \;.
	\|u-\tu\| \leq \sqrt{\frac{4\|M-S\|_2}{\lambda-\lambda_2}}\;.
\end{eqnarray*}
In particular, for $\|M-S\|_2 \leq \theta^2\|M\|_2/(2(1-l_2))$ as per our assumption, 
this is less than $\sqrt{2}\theta/(1-l_2)$.

To bound the second term, we use $x^{(0)}\in\cF_n$ to get   
\begin{eqnarray*}
 	\frac{(\tu^\ast S^kx^{(0)})^2}{\|S^kx^{(0)}\|^2} %&\geq& \frac{\tlam_1^k|\tu_1^\ast x_0|}{\sum_{i\in[n]}\tlam_i^k|\tu_i^\ast x_0|} \\
	&\geq& \frac{1}{1+\sum_{i\geq2}\frac{\tlam_i^{2k}}{\tlam_1^{2k}}\frac{(\tu_i^\ast x^{(0)})^2}{(\tu_1^\ast x^{(0)})^2}}\\
	&\geq& 1- (\tlam_2/\tlam_1)^{2k} 6n^5\log n\\
	%&\geq& 1-\frac18\left(\frac{19}{20}\right)^2\left(1-l_2^2\right) \;,
	&\geq& 1-\frac{\theta^2}{4(1-l_2)^2} \;.
\end{eqnarray*}
In the last inequality we used $k\geq 3\log(n/\theta)/(1-l_2-\theta)$,  
and the fact that $(\tlam_2/\tlam_1)\leq l_2+\theta$. 
Then, $\|\tu-\bx^{(k)}\| \leq {\theta}/({\sqrt{2}(1-l_2)})$. 
Collecting both terms, this proves the desired claim. 
%\begin{eqnarray*}
	%k\geq \frac{\log(8\cdot40^2\sqrt{6}/19^2)+3\log(n)+\frac{1}{2}\log(\log n)-\log(1-l_2^2)}	{-\log(1-(37/40)(1-l_2))} \;, 
	%k\geq \frac{\log(20\sqrt{6})+3\log(n)+\frac{1}{2}\log(\log n)-(1/2)\log(1-l_2^2)}
	%{-\log(1-(37/40)(1-l_2))} \;, 
%\end{eqnarray*}
%which is true for $k\geq k_1$.
\endproof

%
%*************************************************************************
%
\subsection{Proof of Lemma \ref{lem:goodset}}
\label{app:goodset}

\begin{lemma}[Contraction]
	\label{lem:contract}
	%Let $Q$ be any matrix such that $\|Q-M\|_2\leq \theta \lambda$. 
	For a given $\nu\leq(1/20)$, 
	assume that $x,x'$ satisfy $\|x\|=\|x'\|=1$, $\<u,x\>\geq0$, $\<u,x'\>\geq0$,  
	$\|\cP_{u_\perp}(x)\|\leq \nu$, and $\|\cP_{u_\perp}(x')\|\leq \nu$. 
	Then, under the hypothesis of Lemma \ref{lem:goodset}, we have 
	\begin{align}
	\label{eq:contractperpu}
	\left\|\cP_{u_\perp}\left(\frac{Qx}{\|Qx\|}-\frac{Qx'}{\|Qx'\|}\right)\right\| &\leq 
	\big(l_2(1+3\nu^2) + 3\theta\big) \|z-z' \| \;, 
	\end{align}
	and 
	\begin{align}
	\label{eq:contractalongu}
	\left\|\cP_{u}\left(\frac{Qx}{\|Qx\|}-\frac{Qx'}{\|Qx'\|}\right)\right\| &\leq 
	\big(4\nu + 4\theta\big) \|z-z'\|\;,
	\end{align}
	where $l_2\equiv |\lambda_2|/\lambda$, $z=\cP_{u_\perp}(x)$ and $z'=\cP_{u_\perp}(x')$.
\end{lemma}

\begin{proof}

By the assumption that $\<u,x\>\geq0$ and $\<u,x'\>\geq0$, 
we have $x=\sqrt{1-\|z\|^2}u+z$ and $x'=\sqrt{1-\|z'\|^2}u+z'$.
The following inequalities, which follow from $\|Q-M\|_2\leq\theta\lambda$, 
will be frequently used.  
\begin{eqnarray*}
	(1-\theta)\lambda\;\leq&\|Qu\|&\leq\;(1+\theta)\lambda\;,\\
	(l_2-\theta)\lambda\;\leq&\|Qz\|&\leq\;(l_2+\theta)\lambda\;.
\end{eqnarray*}
The following inequalities will also be useful in the proof.
\begin{align}
	\label{eq:uppernorm}
 	\|x-x'\| \leq (1/\sqrt{1-\nu^2})\|z-z'\|\;,
\end{align}
where we used 
$\sqrt{1-a^2}-\sqrt{1-b^2}\leq (\nu/\sqrt{1-\nu^2}) |a-b|$ for 
$|a|\leq\nu$ and $|b|\leq\nu$. Similarly, using the fact that 
$Mu$ and $Mz$ are orthogonal
\begin{align}
	\label{eq:lowernorm}
 	\|Qx\|&\geq (\sqrt{1-\|z\|^2})\|Mu\|-\|Q-M\|_2 \nonumber\\
	&\geq \lambda(\sqrt{1-\nu^2}-\theta)\;,
\end{align}
Next, we want to show that 
\begin{align}
	\label{eq:uppernorm2}
	\left|\frac{1}{\|Qx\|}-\frac{1}{\|Qx'\|}\right| \leq 
	\frac{(2.2\nu+0.1\theta)}{(\sqrt{1-\nu^2}-\theta)^3}\|z-z'\|\;.
\end{align}
We use the equality $1/a-1/b=(a^2-b^2)/(ab(a+b))$ with
$a={\|Qx\|}$ and $b=\|Qx'\|$.
The denominator can be bounded using \eqref{eq:lowernorm}. 
It is enough to bound $|\|Qx'\|^2-\|Qx\|^2|$ using 
\begin{align*}
&\left| \|Q(\sqrt{1-\|z\|^2} u + z)\|^2 - \|Q(\sqrt{1-\|z'\|^2} u + z')\|^2 \right| \\
&\leq \Big|\|z\|^2-\|z'\|^2\Big|\|Qu\|^2  + \Big|\|Qz\|^2-\|Qz'\|^2\Big|  \\
&\;\;\;+ 2\Big|(\sqrt{1-\|z\|^2}z-\sqrt{1-\|z'\|^2}z')^\ast Q^\ast Qu\Big|\;.
\end{align*}
Note that $\Big|\|z\|^2-\|z'\|^2\Big|\|Qu\|^2 \leq 2\nu(1+\theta)^2\lambda^2\|z-z'\|$, 
and $\big|\|Qz\|^2-\|Qz'\|^2\big| \leq 2\nu(l_2+\theta)^2\lambda^2\|z-z'\|$. 
The last term can be decomposed into
\begin{align*}
&2\Big|(\sqrt{1-\|z\|^2}z-\sqrt{1-\|z'\|^2}z')^\ast Q^\ast Qu\Big|\\
&\leq 2|\sqrt{1-\|z\|^2}-\sqrt{1-\|z'\|^2}|\,|z^\ast Q^\ast Qu|\\
& + 2\sqrt{1-\|z'\|^2}\,|(z-z')^\ast Q^\ast Qu|\;.
\end{align*}
Note that $|\sqrt{1-\|z\|^2}-\sqrt{1-\|z'\|^2}|\leq(\nu/\sqrt{1-\nu^2})\|z-z'\|$,  
$|z^\ast Q^\ast Qu|\leq\lambda^2\theta(l_2+\theta)$, 
and $|(z-z')^\ast Q^\ast Qu|\leq \lambda^2(l_2+\theta)\theta\|z-z'\|$.
Collecting all the terms and assuming $\theta\leq1/40$ and $\nu\leq1/20$, 
$|\|Qx\|-\|Qx'\|| \leq (4.4\nu+0.1\theta)\lambda^2\|z-z'\|$.
this implies \eqref{eq:uppernorm2}. 

To prove \eqref{eq:contractperpu}, define $T_1\equiv \cP_{u^\perp}(Qx-Qx')/\|Qx\|$ and 
$T_2\equiv\cP_{u^\perp}(Qx')\big((1/\|Qx\|)-(1/\|Qx'\|)\big)$. 
We bound each of these separately. 
\begin{align*}
\|T_1\| &= \frac{\left\|\cP_{u^\perp}\big(M(x-x') + (Q-M)(x-x')\big)\right\|}{\|Qx\|}\\
	&\stackrel{(a)}{\leq} \frac{l_2\|z-z'\| + \theta\|x-x'\|}{(\sqrt{1-\nu^2}-\theta)}\\
 	&\stackrel{(b)}{\leq} \frac{l_2+(\theta/\sqrt{1-\nu^2})}{(\sqrt{1-\nu^2}-\theta)}\|z - z'\|,
\end{align*}
where $(a)$ follows from \eqref{eq:lowernorm} and the fact that $\cP_{u^\perp}Mu=0$, 
and $(b)$ follows from \eqref{eq:uppernorm}.
Similarly, using \eqref{eq:uppernorm2} 
\begin{align*}
\|T_2\|&=\Big\|\cP_{u^\perp}\big(Q(\sqrt{1-\|z'\|^2}u+z')\big)\Big\| 
	\Big|\frac{1}{\|Qx\|}-\frac{1}{\|Qx'\|} \Big|\\
&\leq(\theta +\nu l_2)\frac{2.2\nu+0.1\theta}{(\sqrt{1-\nu^2}-\theta)^3}\|z-z'\|\;.
\end{align*}
%\begin{align*}
%\|T_2\|&=\Big\|\cP_{u^\perp}\big(Q(\sqrt{1-\|z'\|^2}u+z')\big)\Big\| 
%	\Big(\frac{1}{\|Qx\|}-\frac{1}{\|Qx'\|} \Big)\\
%&\leq\big(\|Q-M\| + \lambda_2 \|z'\|\big)\frac{\|Q(\sqrt{1-\|z\|^2} u + z - \sqrt{1-\|z'\|^2} u 
%	- z')\|}{\lambda^2(\sqrt{1-\nu^2}-\theta)}\\
%&\leq(\theta + \nu l_2)\frac{ (\theta +1) \big|\sqrt{1-\|z\|^2} - \sqrt{1-\|z'\|^2}\big| +
%	(\theta + l_2)\|z - z'\|}{(\sqrt{1-\nu^2}-\theta)^2}\\
%&\leq(\theta +\nu l_2) \frac{\frac{(\theta+1)\nu}{\sqrt{1-\nu^2}}+ \theta+l_2}{(\sqrt{1-\nu^2}-\theta)^2} 
%	\|z - z'\|\;.
%\end{align*}
Notice that by assumption, we have $\theta\leq(1/40)$, 
and by the definition of $\G$ in \eqref{eq:defG}, we have $\nu\leq(1/20)$.
Then, after some calculations, we have proved \eqref{eq:contractperpu}.
Analogously we can prove \eqref{eq:contractalongu} by bounding 
$T_3\equiv \cP_u(Qx-Qx')/\|Qx\|$ and 
$T_4\equiv \cP_u(Qx')\big((1/\|Qx\|)-(1/\|Qx'\|)\big)$ separately. 
\end{proof}
%
%------------------------------------------------
%

We are now in position to prove Lemma \ref{lem:goodset}.

{\bf Proof of Lemma \ref{lem:goodset}}. 
% Proof of staying in G
We first show that for any $\bx\in\G$ with a representative $x$ such that $\<x,u\>\geq0$, 
we have $Q\bx\in\G$. 
Note that, by triangular inequality, 
$\|\cP_{u^\perp}(Qu)\|\leq\theta\lambda $ and $\|Qu\|\geq(1-\theta)\lambda$.
Applying Lemma~\ref{lem:contract} to $x$ and $u$, we get 
\begin{align}
	&\left\|\cP_{u_\perp}\left(\frac{Qx}{\|Qx\|}\right)\right\| \nonumber\\
	&\leq 
	 \left\|\cP_{u_\perp}\left(\frac{Qu}{\|Qu\|}\right)\right\| + \Big(l_2\Big(1+3\nu^2\Big)+3\theta\Big)
	 \left\|\cP_{u_\perp}(x-u)\right\|\nonumber\\
	&\leq \Big(\frac{1}{1-\theta}+3\nu\Big)\theta + \Big(l_2\Big(1+3\nu^2\Big)\Big)\nu \;.\label{eq:StayInG}
\end{align}
For $\theta\leq(1/40)$ and for $\theta$ and $\nu$ satisfying, 
\begin{align*}
	\frac{2}{1-l_2}\theta\leq \nu \leq \min\Big\{\sqrt{\frac{2(1-l_2)}{15}}\;,\;\frac{1}{20}\Big\}\;,
\end{align*}
the right-hand side of \eqref{eq:StayInG} is always smaller than $\nu$, 
since $ \big((1/(1-\theta))+3\nu\big)\theta \leq (3/5)(1-l_2)\nu$ and 
$3\nu^2 \leq (2/5)(1-l_2) $. 
This proves our claim for $\theta\leq(1/40)(1-l_2)^{3/2}$ 
and $\nu \in [(2\theta)/(1-l_2),\sqrt{1-l_2}/20]$ as per our assumptions.

% Proof of Contraction in G
Next, we show that there is a contraction in the set $\G$.
For $x$ and $x'$ satisfying the assumptions in Lemma \ref{lem:contract}, 
define $y\equiv{Ax}/{\|Ax\|}$, $y'\equiv{Ax'}/{\|Ax'\|}$, 
$z\equiv\cP_{u^\perp}(x)$, and $z'\equiv\cP_{u^\perp}(x')$. 
For $\|x\|\leq\nu$ and $\|x'\|\leq\nu$ we have
\begin{align*}
	1-2\nu^2\leq (x,x')\leq 1.
\end{align*}
Using the above bounds we get 
\begin{align*}
\frac{1-\<y,y'\>^2}{1-\<x,x'\>^2} \leq \frac{1}{(1-\nu^2)}\frac{1-\<y,y'\>}{1-\<x,x'\>} \;.
\end{align*}
We can further bound $(1-\<y,y'\>)/(1-\<x,x'\>)$ using Lemma \ref{lem:contract}. 
\begin{align*}
\|y - y'\| \leq \sqrt{\left(l_2+3\nu^2+3\theta\right)^2+\left(4\nu+4\theta\right)^2}\,\|z-z'\|  \;.
\end{align*}
Using $\|z-z'\|^2\leq\|x-x'\|^2=2-2\<x,x'\>$, we get
\begin{align*}
\frac{1-\<y,y'\>}{1-\<x,x'\>} \leq  \left(l_2+3\nu^2+3\theta\right)^2+\left(4\nu+4\theta\right)^2 \;.
\end{align*}
For $\theta\leq(1/40)(1-l_2)^{3/2}$ and $\nu \leq \sqrt{1-l_2}/20$ as per our assumptions, 
%Recall that, by assumption, we have $\theta\leq(1/40)$ and $\theta\leq(1/20)$. 
it follows, after some algebra, that 
\begin{align*}
	\sqrt{\frac{1-\<y,y'\>^2}{1-\<x,x'\>^2}} \leq \sqrt{\frac{\left(l_2+3\nu^2+3\theta\right)^2+\left(4\nu+4\theta\right)^2}{1-\nu^2}} \leq \rho \;,
\end{align*}
for $\rho \geq 1-0.8(1-l_2)$. 

%we can show that 
%the right-hand side of the above inequality is less than $(1-\nu^2)$ provided that 
%$24.1 \nu^2 + 8.3 \theta \leq 1-l_2^2$.
%Further, this implies that 
%\begin{align*}
%	\sqrt{\frac{1-\<y,y'\>^2}{1-\<x,x'\>^2}} \leq 1-\delta \;,
%\end{align*}
%for $\delta \leq (1/2)-(l_2^2+23.1\nu^2+8.3\theta)/(2(1-\nu^2))$.
%For $\theta\leq(1/40)(1-l_2)^{3/2}$ 
%and $\nu \in [(2\theta)/(1-l_2),\sqrt{1-l_2}/20]$ as per our assumptions,  
%this proves the desired claim for $\delta \leq 0.3(1-l_2)$. 

\endproof

%
%=========================================================================
%
\subsection{Proof of Lemma \ref{lem:NormConvergence}}
\label{app:NormConvergence}

Expanding the summation, we get  
\begin{align*}
&\Big\|\frac{1}{t}\sum_{s=1}^{t} f(\bX_s) - \stat(f)\Big\|^2 \\
%&\;\;\;= \frac1{t^2}\sum_{r,s=1}^{t} \big\<f(\bX_r)-\stat(f),f(\bX_s)-\stat(f)\big\>\\
&\;\;\;= \frac{1}{t^2}\sum_{s=1}^{t}\|f(\bX_s)-\stat(f)\|^2 + \\
&\;\;\;\;\;\;\frac{2}{t^2}\sum_{r=1}^{t}\sum_{r<s}\big\<f(\bX_r)-\stat(f),f(\bX_s)-\stat(f)\big\> \;,
\end{align*}
where $\<a,b\>=a^\ast b$ denotes the scalar product of two vectors.
We can bound the first term by $20\theta^2/(t(1-l_2)^2)$, 
since 
\begin{eqnarray}
	\label{eq:mubound}
 	\|f(\bX_s)-\mu(f)\|^2\leq \frac{20\theta^2}{(1-l_2)^2} \;,
\end{eqnarray}
where we used $\|\cP_u(f(\bX_s)-\mu(f))\|^2\leq 4\theta^2/(1-l_2)^2$ 
and $\|\cP_{u^\perp}(f(\bX_s)-\mu(f))\|^2
\leq 16\theta^2/(1-l_2)^2$ for $\bX_s\in\G$. 

To bound the second term, let $y\equiv f(\bX_r)-\stat(f)$. 
Note that by \eqref{eq:mubound}, $\|y\|\leq\sqrt{20}\theta/(1-l_2)$. 
We apply Theorem \ref{thm:GeneralConv} together with Lemma \ref{lem:goodset} to get 
\begin{eqnarray*}
\Big|\bE\big[y^\ast f(\bX_s)\,\big|\,\bX_r\big]-y^\ast\mu(f) \Big| \leq\rho^{s-r}\|y\|^2 \;, 
\end{eqnarray*}
for $r < s$. 
Using the fact that for $|\rho|\leq1$, $\sum_{r=1}^{t}\sum_{r< s} \rho^{s-r} \leq 
\sum_{r=1}^{t} \rho^{-r}\rho^{r+1}/(1-\rho)\leq \rho t/(1-\rho)$, it follows that 
\begin{align*}
&\sum_{r=1}^{t}\sum_{r<s} \bE\Big[\big\<f(\bX_r) -\stat(f), f(\bX_s)-\stat(f)\big\>\Big] \\
&\;\;\;\;\leq \sum_{r=1}^{t}\sum_{r< s} \bE\Big[\big\<f(\bX_r) -\stat(f), \bE\big[f(\bX_s)-\stat(f)\,\big|\,\bX_r\big]\big\>\Big] \\ 
&\;\;\;\;\leq \frac{20\theta^2}{(1-l_2)}\sum_{r=1}^{t}\sum_{r< s} \rho^{s-r} \leq \frac{20\theta^2 \rho t}{(1-l_2)(1-\rho)} \;.
\end{align*}

Combining the above bounds we get
\begin{align*}
\bE\Big\|\frac1t\sum_{s=1}^t f(\bX_s) - \stat(f)\Big\|^2 \leq \frac{20\theta^2}{t(1-l_2)^2}+\frac{40\theta^2\rho}{(1-l_2)(1-\rho) t}.
\end{align*}
For $\rho=1-(4/5)(1-l_2)$ as in Lemma \ref{lem:goodset}, this proves the desired claim.
\endproof

%
%=========================================================================
%
\subsection{Proof of Lemma \ref{lem:statdist}}
\label{app:statdist}

From Theorem \ref{thm:muproperty}, we know $\mu(\G^c)=0$. 
This imples that $\|\mu(f)\|^2 \geq \|\cP_u(\mu(f))\|^2 \geq 1-1/400$. 
Then,
\begin{align*}
	d(\bu,\mu(f)) = \frac{\|\cP_{u^\perp}(\stat(f))\|}{\|\mu(f)\|} \leq 2\|\cP_{u^\perp}(\stat(f))\|\;.
\end{align*}

Let $\bX$ be an random element in $\P_n$ following 
the stationary distribution $\stat(\cdot)$, and 
the random vector $X\in\reals^{n}$ be the representative. 
From the definition of $f(\cdot)$, 
$\cP_{u^\perp}(X)$ is invariant when we apply $f(\cdot)$, 
whence $\cP_{u_\perp}(f(\bX))=\cP_{u_\perp}(X)$. 
We can bound $\|\cP_{u_\perp}(X)\|$ with the following recursion.
\begin{align*}
&\cP_{u_\perp}(X) = 
\bE\left[\frac{\cP_{u_\perp}(\,\prod_{\ell=1}^{k}S^{(\ell)}X)}{\|\prod_{\ell=1}^{k}S^{(\ell)}X\|}\,\Big|\,X\right]\\
& = \bE\left[ \cP_{u_\perp}\Big(\prod_{\ell=1}^{k}S^{(\ell)}X\Big)
   \Big(\frac{1}{\|\prod_{\ell=1}^{k} S^{(\ell)} X\|}- \frac{1}{\|M^kX\|}\Big) \,\Big|\,X\right]  \\
&\;\;\;\;+ \frac{\bE\left[\cP_{u_\perp}\Big(\,\prod_{\ell=1}^{k} S^{(\ell)}X\Big)\,\big|\,X\right]}{\|M^k X\|}\\
& = \bE\Big[ \cP_{u_\perp}\Big(\prod_{\ell=1}^{k}S^{(\ell)}X\Big)\Big(\frac{1}{\|\prod_{\ell=1}^{k}S^{(\ell)}X\|}- \frac{1}{\|M^kX\|}\Big) \,\Big|\,X\Big] \\
&\;\;\;\;+ \frac{\cP_{u_\perp}(M^k X)}{\|M^k X\|}\;.
\end{align*}

Let $\nu\equiv2\theta/(1-l_2)$. 
To bound the second term, note that $\mu(\G^c)=0$.
This imples that $\|X\| \geq \|\cP_u(X)\|\geq \sqrt{1-\nu^2}$ 
and $\cP_{u^\perp}(X)\leq \nu$ with probability one. Then, 
\begin{align*}
\frac{\left\|\cP_{u_\perp}(M^k X)\right\|}{\|M^k X\|} \leq l_2^k \frac{\nu}{\sqrt{1-\nu^2}}, 
\end{align*}
with probability one.
To bound the first term, we use the telescoping sum 
\begin{align*}
 \prod_{\ell=1}^{k}S^{(\ell)}-M^k=\sum_{i=1}^k \Big(\prod_{\ell=i+1}^{k}S^{(\ell)}\Big)(S^{(i)}-M)M^{i-1}\;.  
\end{align*}
Applying the triangular inequality of the operator norm,  
%to get $\Big\|\sum_{i=1}^k \prod_{s=i+1}^{k}M^{(s)}(M^{(i)}-M)M^{i-1}\Big\|_2 \leq \sum_{i=1}^k \|\prod_{s=i+1}^{k}M^{(s)}(M^{(i)}-M)M^{i-1}\|_2$, 
we have 
\begin{align*}
	\Big\|\prod_{\ell=1}^{k}S^{(\ell)}-M^k\Big\|_2 \leq \lambda^k\big((1+\theta)^{k} - 1\big) \;,
\end{align*}
which follows from $ \Big(\prod_{\ell=i+1}^{k}S^{(\ell)}\Big)(M^{(i)}-M)M^{i-1} \leq \lambda^k (1+\theta)^{k-i} \theta $. Using the above inequality, we get the following bounds with probability one.  
\begin{align*}
\Big\|\cP_{u_\perp}\Big(\prod_{\ell=1}^kS^{(\ell)} X\Big)\Big\| 
&\leq \big\|\cP_{u_\perp}\big(M^k X\big)\big\| + \Big\|\big(\prod_{\ell=1}^kS^{(\ell)}-M^k\big)X\Big\|\\
&\leq \lambda^k(l_2^k \nu + (1+\theta)^k -1) \;, \text{ and }\\
\Big\|\cP_u\Big(\prod_{\ell=1}^kS^{(\ell)}X\Big)\Big\| 
&\geq \big\|\cP_u\big(M^k X\big)\big\| - \Big\|\big(\prod_{\ell=1}^kS^{(\ell)}-M^k\big)X\Big\| \\
&\geq \lambda^k(\sqrt{1-\nu^2} - (1+\theta)^k + 1) \;.
\end{align*}
Then it follows that, 
\begin{align*}
&\left|\frac{1}{\|\prod_{\ell=1}^{k}S^{(\ell)}X\|}- \frac{1}{\|M^kX\|}\right| 
\leq \frac{\|(M^k-\prod_{\ell=1}^kS^{(\ell)}) X\|}{\|M^kX\|\|\prod_{\ell=1}^{k}S^{(\ell)}X\|}\\
&\;\;\;\; \leq \frac{((1+\theta)^k-1)}{\lambda^k\sqrt{1-\nu^2}(\sqrt{1-\nu^2} - (1+\theta)^k +1)}.
\end{align*}
Collecting all the terms, we get 
\begin{align*}
d(\bu,\mu(f)) \leq 
\frac{2((1+\theta)^k -1 + l_2^k \nu)((1+\theta)^k-1)} 
{\sqrt{1-\nu^2}(\sqrt{1-\nu^2} - (1+\theta)^k +1)} 
+  \frac{2l_2^k\nu}{\sqrt{1-\nu^2}} \;.
\end{align*}

Let $k=\lceil \log(\theta)/\log(l_2)\rceil$ such that $l_2^k\leq\theta$. 
From the assumption that $\theta\leq(1-l_2)^{3/2}/40$, it follows that 
$\theta\log \theta \leq 0.12(1-l_2)$. 
Then, 
\begin{align*}
 (1+\theta)^{(\log\theta/\log l_2)}-1 &\leq e^{(\theta\log\theta/\log l_2)}-1 \\
&\leq \frac{1.1}{(1-l_2)}\theta\log(1/\theta)\;,
\end{align*}
Then, after some algebra, $(1+\theta)^k-1 \leq \theta+(1+\theta)((1+\theta)^{(\log\theta/\log l_2)}-1)\leq 1.5\theta\log(1/\theta)/(1-l_2)$, and 
$(1+\theta)^k-1 +l_2^k\nu \leq1.5\theta\log(1/\theta)/(1-l_2)$. 
It also follows that $\sqrt{1-\nu^2} \geq \sqrt{399/400}$ 
and $(\sqrt{1-\nu^2} - (1+\theta)^k +1)\geq 0.8$. 
Collecting all the terms, we get the desired bound on $d(\bu,\mu(f))$.
%\begin{align*}
% d(\bu,\mu(f)) \leq 8\Big(\frac{\theta\log(1/\theta)}{(1-l_2)}\Big)^2\;. 
%\end{align*}

\endproof

%
%=========================================================================
%

\section{Eigenvalue Estimation}
\label{sec:value}

In the previous sections, we discussed 
the challenging task of computing the largest eigenvector 
under the gossip setting.
A closely related task of computing the largest eigenvalue is also 
practically important in many computational problems.
For example, positioning from pairwise distances 
requires the leading eigenvalues, as well as the leading eigenvectors, to correctly find the positions \cite{OKM10}. 
%In this section, we discuss this challenging task of  
%finding the leading eigenvalues 
%under stringent computations and communication constraints. 
In the following, we present an algorithm to estimate the leading eigenvalue under the gossip setting 
and provide a performance guarantee.  
Although the proposed algorithm uses the same trajectory $\{x^{(t)}\}$ from \gossip_pca, 
the analysis is completely different from that of the eigenvector estimator. 

We assume to have at our disposal 
the random trajectory $\{x^{(t)}\}_{t\geq0}$, defined as in \eqref{eq:Basic}, 
possibly from running \gossip_pca. 
Assume that we start with $x^{(0)}$ with entries distributed as $\cN(0,1)$. 
%We iteratively apply i.i.d. sparsifications $\{S^{(\ell)}\}$ 
%to get $x^{(t)}=\prod_{\ell=1}^{t}S^{(\ell)}x^{(0)}$. 
Our estimate for the top eigenvalue $\lambda$ after $t$ iterations is  
\begin{eqnarray*}
 \hlam^{(t)} = \Big\{ \big| \<x^{(0)},x^{(t)}\> \big| \Big\}^{1/t} \;. 
\end{eqnarray*}
Although, this estimator uses  
the same trajectory $\{x^{(t)}\}_{t\geq0}$ as \gossip_pca, 
the analysis significantly differs from that of the eigenvector estimator. 
Hence, the statement of the error bound in 
Theorem \ref{thm:eigenvalue} is 
also significantly different from Theorem \ref{thm:eigenvector}.
The main idea of our analysis is to bound the second moment of the estimate 
and apply Chebyshev's inequality. 
Therefore, the second moment of $S^{(\ell)}_{ij}$ characterized by $\alpha$ 
determines the accuracy of the sparsification. 
\begin{eqnarray}
	\max_{i,j}\var(S^{(\ell)}_{ij}) \leq (\alpha/n) \|M\|_2^2 \;. \label{eq:AlphaDef}
\end{eqnarray}
The trade-off between $d$, which determines the complexity, 
and $\alpha$ depends on the specific sparsification method. 
With random sampling described in Section \ref{sec:example}, 
it is not difficult to show that Eq.~\eqref{eq:AlphaDef} holds with only $\alpha=O(1/d)$. 
%Then, the algorithm design amounts to the choice of number of iterations $t$ 
%and the quality of sparsification, which is characterized by $\alpha$ and $d$. 

Our main result bounds the error of the algorithm in terms of $\alpha$, $t$, and 
$\srank \equiv \sum_{i=1}^n({|\lambda_i|}/{\lambda})$.
The proof of this theorem is outline in the following section. 
\begin{thm}
\label{thm:eigenvalue}
	Let $\{S^{(\ell)}\}_{\ell\geq 1}$ be a sequence of i.i.d. $n\times n$ random matrices satisfying 
	$E[S^{(\ell)}]=M$ and Eq.~\eqref{eq:AlphaDef}.
	Assume $\alpha<1/2$ and $\max\{\log_2 n,2\log_{(1/l_2)} n\}\le t \le n/(4\alpha\srank)$.
	Then with probability larger than $1-\max\{\delta,16/n^2\}$ 
	\begin{align*}
	\left|\frac{\hlam^{(t)}-\lambda}{\lambda}\right|
	\le 
	\max\left\{ \frac{8\sqrt{2}}{tn\sqrt{\delta}} ; \;
	32\sqrt{\frac{\alpha\srank^{3/2}\log n}{t^2\delta}};\;
	48\sqrt{\frac{\alpha\srank^{3}(\log n)^2}{t n\delta}}
	\right\}\, ,
	\end{align*}
	provided the right hand side is smaller than $1/t$.
\end{thm}

%
%=========================================================================
%

\subsection{Proof of Theorem \ref{thm:eigenvalue}}
\label{sec:proof1}

The proof idea is fairly simple. 
Let $x_0=x^{(0)}$ and define 
$\hlt[t] \equiv x_0^\ast x^{(t)}$ such that 
our estimator is $\hlam^{(t)} \equiv (|\hlt[t]|)^{1/t}$. 
We will show that this is close to the desired result $\lambda$
by applying Chebyshev inequality to $\hlt[t]$. In order to do this
we need to compute its mean and variance.
\begin{lemma}\label{lemma:GeneralMoments}
Consider the two operators
$\cA,\cB:\reals^{n\times n}\to\reals^{n\times n}$,
defined as follows
\begin{eqnarray}
\cA(X) &\equiv & MXM^*\, ,\label{eq:defA}\\  
\cB(X) &\equiv & \frac{\lambda^2\alpha}{n} \<X,\id_n\>\, \id_n\, ,\label{eq:defBnew}
\end{eqnarray}
where $\<X,Y\>=\Trace(X^\ast Y)$. Then, conditional on $x_0$ we have
\begin{align*}
\E[\lambda^{(t)}|x_0] &=  x_0^*M^tx_0\, ,\\
\Var(\lambda^{(t)}|x_0) &\leq \<x_0x_0^*,(\cA+\cB)^t(x_0x_0^*)\>-\<x_0x_0^*,\cA^t(x_0x_0^*)\>\, .
\end{align*}
\end{lemma}
The next lemma provides a bound on the variance.
\begin{lemma}\label{lemma:BoundMatrElem}
Let $\cA, \cB:$ $\reals^{n\times n}\to \reals^{n\times n}$ be
defined as in Eqs.~(\ref{eq:defA}) and (\ref{eq:defBnew}). 
Further assume $\alpha<1/2$ and $\alpha t\srank<n/4$. Then,
for any two vectors $x$, $y\in\reals^n$, $\|x\|=\|y\|=1$,
\begin{align*}
&\big|\<yy^*,(\cA+\cB)^t(xx^*)\>-\<yy^*,\cA^t(xx^*)\>\big| \\
&\le 4\lambda^{2t}\Big\{\frac{n\alpha^t+8\alpha^2\srank}{4n^2}
+ \frac{\alpha\sqrt{\srank}}{n} \Big(\sum_{i=1}^n
\frac{|\lambda_i|}{\lambda}\big((u_i^*x)^2+(u_i^*y)^2\big)\Big)\\
&\;\;\;\;+\frac{\alpha t\srank}{n}  \Big(\sum_{i=1}^n
\frac{|\lambda_i|}{\lambda}(u_i^*x)^2\Big)\Big(\sum_{i=1}^n
\frac{|\lambda_i|}{\lambda}(u_i^*y)^2\Big)
\Big\}\, .
\end{align*}
\end{lemma}
For the proof of Lemmas \ref{lemma:GeneralMoments} and \ref{lemma:BoundMatrElem} we refer to the longer version of this paper.
%We are now in position to prove Theorem \ref{thm:eigenvalue}.
%
%\begin{proof}[Proof of Theorem \ref{thm:eigenvalue}]
Let $\good_n$ be any measurable subset of $\reals^n$.
%the unit sphere $S^n = \{x\in\reals^n:\, \|x\|=1\}$. 
This forms a set of `good' initial 
condition $x_0$, and its complement will be denoted by $\nogood_n$.
With an abuse of notation, $\good_n$ will also 
denote the event $x_0\in\good_n$ (and analogously for $\nogood_n$).
Also let, $\hlam=\hlam^{(t)}$. Then, for any $\Delta>0$,
\begin{align*}
&\prob\big\{\hlam\not\in [\lambda(1-\Delta)^{1/t},
\lambda(1+\Delta)^{1/t}]\big\}\\
%&=\prob\big\{\big|\hlam^t-\lambda^t\big|\ge\Delta\lambda^t\big\}\\
&\le \prob\{\big\{\big|\hlam^t-\lambda^t\big|\ge\Delta\lambda^t\,\big|\,
\good_n\big\}+\prob\big\{\nogood_n\big\}\\
&\le \frac{1}{\Delta^2\lambda^{2t}}
\E\big\{(\hlam^t-\lambda^t)^2\big|\,\good_n\big\}+
\prob\big\{\nogood_n\big\}\\
&\le
 \frac{1}{\Delta^2\lambda^{2t}}
\big(\E\{\hlam^t\big|\,\good_n\}-\lambda^t\big)^2+
 \frac{1}{\Delta^2\lambda^{2t}}\sup_{x_0\in\good_n}
\var\big(\hlam^t\big|\,x_0\big)+
\prob\big\{\nogood_n\big\}\, .
\end{align*}
We  shall upper bound each of the three terms in the above expression with 
\begin{eqnarray}
\good_n\equiv \Big\{ x\in \reals^n\,:\, \max_{i\le 1}|u_i^*x|\le \sqrt{6\log n}\Big\}\, ,
\end{eqnarray}
%

%Notice that $(u_i^*x_0)$ converges for large $n$ to $\normal(0,1/n)$,
%whence $\prob\{(u_i^*x_0)^2\ge (6\log n)/n\} = 1/n^3$ \cite{Ledoux}. 
Notice that $\prob\{(u_i^*x_0)^2\ge 6\log n\}\leq2/n^3$. 
By the union bound we get $\prob\{x_0\in \nogood_n\}\le {2}/{n^2}$.

Next observe that
\begin{align*}
\frac{1}{\lambda^t}\E\{\hlam^t\big|\,\good_n\}-1
%= \frac{n}{\lambda^t} \E\{x_0^*M^tx_0|\good_n\}-1
=\big(\E\{(u_1^*x_0)^2|\good_n\}-1\Big)+\sum_{i=2}^n\frac{\lambda_i^t}{\lambda^t}(u_i^*x_0)^2\, .
\end{align*}
The first step can be computed as
\begin{align*}
\E\{(u_1^*x_0)^2|\good_n\}= \frac{\E\{(u_1^*x_0)^2\}
-\E\{(u_1^*x_0)^2\,\ind_{\nogood_n}\}}{1-\prob\{x_0\in\nogood_n\}}
\, ,
\end{align*}
whence, recalling that $\E\{(u_1^*x_0)^2\}=1$, and
$\prob\{x_0\in\nogood_n\}\le 1/2$ for 
all $n$ large enough, we get
\begin{align*}
\big|\E\{(u_1^*x_0)^2|\good_n\}-1\big|
&= \left|\frac{\prob\{x_0\in\nogood_n\}-\E\{(u_1^*x_0)^2\,\ind_{\{x_0\in\nogood_n\}}\}}
{1-\prob\{x_0\in\nogood_n\}}\right|\\
&\le \frac{4}{n^2}\, .
\end{align*}
%
%{\bf [Note further that $\sum_{i=1}^n(u_i^*x_0)^2=1$, since each entry of $x_0$ is independent and $+1/\sqrt{n}$ or $-1/\sqrt{n}$ with probability half.] }
Note further that, by Chernoff inequality, 
$\sum_{i=1}^n(u_i^*x_0)^2\leq 3n$ with probability at least $1-\exp\{(1/10)n\}$. Then, 
% This follows from $\E[e^{\lambda x_{0,i}^2}] = \sqrt{2n/(n-2\lambda)}$ 
% for $\lambda=(n/3)<(n/2)$, and applying Chernoff inequality, 
% \begin{eqnarray*}
% \prob(\|x_0\|^2\geq 3) &\leq& e^{-3\lambda} \prod_{i\in[n]} \E[e^{\lambda x_{0,i}^2}] \\
% 		     &\leq& \exp\left\{ -n + n\log\Big(\sqrt{\frac{2n}{n-2\lambda}}\Big)\right\}\\
%  	             &\leq& \exp \left\{ -\frac{1}{10}n\right\}\;,
% \end{eqnarray*}
%
\begin{align*}
\Big|\frac{1}{\lambda^t}\E\{\hlam^t\big|\,\good_n\}-1\Big|\le \frac{4}{n^2}+ 3n\Big(\frac{|\lambda_2|}{\lambda}\Big)^t\, .
\end{align*}

Finally, using Lemma \ref{lemma:GeneralMoments} and \ref{lemma:BoundMatrElem} we get
\begin{align*}
\frac{1}{\lambda^{2t}}\,\Var(\hlam^t|x_0) &\le 
4\, n^2\, \Big\{\frac{\alpha^t}{4n}+\frac{2\alpha^2\srank}{n^2}
+ 2\frac{\alpha\sqrt{\srank}}{n} \sum_{i=1}^n
\frac{|\lambda_i|}{\lambda}(u_i^*x_0)^2\\
&\;\;+\frac{\alpha t\srank}{n}  \Big(\sum_{i=1}^n
\frac{|\lambda_i|}{\lambda}(u_i^*x_0)^2\Big)^2
\Big\}\;.
\end{align*}
Further, for any $x_0\in\good_n$, 
$\sum_{i=1}^n({|\lambda_i|}/{\lambda})(u_i^*x_0)^2\le 6\, \srank \log n$, and therefore
\begin{align*}
&\frac{1}{\lambda^{2t}}\,\Var(\hlam^t|x_0) \\
& \le 4n\Big\{\frac{\alpha^t}{4}+\frac{2\alpha^2\srank}{n}
+\frac{12\alpha\srank^{3/2}\log n}{n} + \frac{36\alpha t\srank^3(\log n)^2}{n^2}\Big\}\\
&\le 
4n\Big\{\frac{3\alpha^2\srank}{n}
+\frac{12\alpha\srank^{3/2}(\log n)}{n} + \frac{36\alpha t\srank^3(\log n)^2}{n^2}
\Big\}\\
&\le 
4\Big\{15\alpha\srank^{3/2}(\log n) + \frac{36\alpha t\srank^3(\log n)^2}{n}
\Big\}
\, ,
\end{align*}
where we used the fact that, for $\alpha<1/2$ and $t\geq\log_2(n)$, we 
have $\alpha^t/4 \le \alpha^2\srank/n$.

Collecting the various terms we obtain
\begin{align*}
&\prob\big\{\hlam\not\in [\lambda(1-\Delta)^{1/t},\lambda(1+\Delta)^{1/t}]\big\} \\
&\le 
\frac{2}{n^2}+\frac{1}{\Delta^2}\Big\{\frac{4}{n^2}+nl_2^t\Big\}^2
+\frac{4\alpha\log n}{\Delta^2}\Big\{15\srank^{3/2} + 
\frac{36 t\srank^3(\log n)}{n}\Big\}\, ,
\end{align*}
whence $\prob\big\{\hlam\not\in [\lambda(1-\Delta)^{1/t},
\lambda(1+\Delta)^{1/t}]\big\}\le \delta$,
provided $n \ge 4/\sqrt{\delta}$, $\Delta\ge 4\sqrt2/(n\sqrt{\delta})$, 
$\Delta\ge 2n(|\lambda_2|/\lambda)^t /\sqrt{\delta}$,
$\Delta^2 \ge (240\alpha\srank^{3/2}\log n)/\delta$,
$\Delta^2 \ge 4\cdot 144\alpha t\srank^3(\log n)^2/(n\delta)$. 
The thesis follows by noting that $(1+\Delta)^{1/t}\le 1+(\Delta/t)$
and $(1-\Delta)^{1/t}\ge 1-2(\Delta/t)$ provided $\Delta\le 1/2$.

%\end{proof}

%\section*{Acknowledgements}
%The work of Satish Korada was supported by a Swiss National Science Foundation fellowship.

\bibliographystyle{abbrv}
\bibliography{GossipPCA}

\begin{thebibliography}{10}

\bibitem{Ekahau}
Ekahau.
\newblock {\sf http://www.ekahau.com}.

\bibitem{Qwikker}
Qwikker.
\newblock {\sf http://qwikker.com}.

\bibitem{Sonitor}
Sonitor technologies.
\newblock {\sf http://www.sonitor.com}.

\bibitem{AM02}
D.~Achlioptas and F.~McSherry.
\newblock Fast computation of low-rank matrix approximations.
\newblock {\em J. ACM}, 54(2):9, 2007.

\bibitem{BoydGossip}
S.~Boyd, A.~Ghosh, B.~Prabhakar, and D.~Shah.
\newblock Randomized gossip algorithms.
\newblock {\em IEEE Trans. on Inform. Theory}, 52:2508 -- 2530, 2006.

\bibitem{PageRank}
S.~Brin and L.~Page.
\newblock The anatomy of a large-scale hypertextual web search engine.
\newblock {\em Comput. Netw. ISDN Syst.}, 30:107--117, April 1998.

\bibitem{LSI}
S.~Deerwester, S.~T. Dumais, G.~W. Furnas, T.~K. Landauer, and R.~Harshman.
\newblock Indexing by latent semantic analysis.
\newblock {\em Journal of the American Society for information science},
  41(6):391--407, 1990.

\bibitem{DimakisReview}
A.~Dimakis, S.~Kar, J.~Moura, M.~Rabbat, and A.~Scaglione.
\newblock Gossip algorithms for distributed signal processing.
\newblock {\em Proc. of the IEEE}, 98:1847--1864, 2010.

\bibitem{DFKVV99}
P.~Drineas, A.~Frieze, R.~Kannan, S.~Vempala, and V.~Vinay.
\newblock Clustering in large graphs and matrices.
\newblock In {\em SODA '99}, pages 291--299, 1999.

\bibitem{DK03}
P.~Drineas and R.~Kannan.
\newblock Pass efficient algorithms for approximating large matrices.
\newblock In {\em SODA '03}, pages 223--232, 2003.

\bibitem{DKM06}
P.~Drineas, R.~Kannan, and M.~W. Mahoney.
\newblock Fast monte carlo algorithms for matrices ii: Computing a low-rank
  approximation to a matrix.
\newblock {\em SIAM J. Comput.}, 36(1), 2006.

\bibitem{DM05}
P.~Drineas and M.~W. Mahoney.
\newblock On the nystr\"{o}m method for approximating a gram matrix for
  improved kernel-based learning.
\newblock {\em J. Mach. Learn. Res.}, 6:2153--2175, 2005.

\bibitem{Quantized2}
P.~Frasca, R.~Carli, F.~Fagnani, and S.~Zampieri.
\newblock Average consensus by gossip algorithms with quantized communication.
\newblock In {\em 47th IEEE Conference on Decision and Control}, pages
  4831--4836, 2008.

\bibitem{FKV04}
A.~Frieze, R.~Kannan, and S.~Vempala.
\newblock Fast monte-carlo algorithms for finding low-rank approximations.
\newblock {\em J. ACM}, 51(6):1025--1041, 2004.

\bibitem{Furstenberg}
H.~Furstenberg and H.~Kesten.
\newblock Products of random matrices.
\newblock {\em The Annals of Mathematical Statistics}, 31(2):457--469, June
  1960.

\bibitem{TroppReview}
N.~Halko, P.~Martinsson, and J.~A. Tropp.
\newblock Finding structure with randomness: Stochastic algorithms for
  constructing approximate matrix decompositions.
\newblock {\tt arXiv:0909.4061}, 2010.

\bibitem{StatisticalLearning}
T.~Hastie, R.~Tibshirani, and J.~H. Friedman.
\newblock {\em {The Elements of Statistical Learning}}.
\newblock Springer, 2003.

\bibitem{HJ90}
R.~A. Horn and C.~R. Johnson.
\newblock {\em Matrix Analysis}.
\newblock Cambridge University Press, 1990.

\bibitem{Jolliffe}
I.~T. Jolliffe.
\newblock {\em Principal component analysis}.
\newblock Springer-Verlag, 1986.

\bibitem{Quantized1}
A.~Kashyap, T.~Basara, and R.~Srikant.
\newblock Quantized consensus.
\newblock {\em Automatica}, 43:1192--1203, 2007.

\bibitem{KM08}
D.~Kempe and F.~McSherry.
\newblock A decentralized algorithm for spectral analysis.
\newblock {\em Journal of Computer and System Sciences}, 74(1):70 -- 83, 2008.

\bibitem{KMO09}
R.~H. Keshavan, A.~Montanari, and S.~Oh.
\newblock Matrix completion from a few entries.
\newblock {\em IEEE Trans. Inform. Theory}, 56(6):2980--2998, June 2010.

\bibitem{LePage}
M.~{Le Page}.
\newblock Theoremes limites pour les produits de matrices aleatoires.
\newblock {\em Probability Measures on Groups}, 928:258--303, 1982.

\bibitem{OKM10}
S.~Oh, A.~Karbasi, and A.~Montanari.
\newblock Sensor network localization from local connectivity: Performance
  analysis for the mds-map algorithm.
\newblock In {\em Proc. of the IEEE Inform. Theory Workshop}, January 2010.

\bibitem{Oseledets}
V.~Oseledets.
\newblock A multiplicative ergodic theorem. lyapunov characteristic numbers for
  dynamical systems.
\newblock {\em Trans.Moscow Math. Soc.}, 19:197--231, 1968.

\bibitem{GPS}
B.~W. Parkinson and J.~J. Spilker.
\newblock {\em The global positioning system: theory and applications}.
\newblock American Institute of Aeronautics and Astronautics, 1996.

\bibitem{ShahReview}
D.~Shah.
\newblock Gossip algorithms.
\newblock {\em Foundations and Trends in Networking}, 3, 2009.

\bibitem{SpielmanSparse}
D.~Spielman and N.~Srivastava.
\newblock Graph sparsification by effective resistances.
\newblock In {\em 40th annual ACM symposium on Theory of computing}, 2008.

\bibitem{Isomap}
J.~B. Tenenbaum, V.~Silva, and J.~C. Langford.
\newblock {A Global Geometric Framework for Nonlinear Dimensionality
  Reduction}.
\newblock {\em Science}, 290(5500):2319--2323, 2000.

\end{thebibliography}
% sigproc.bib is the name of the Bibliography in this case
% ACM needs 'a single self-contained file'!
%

\newpage

\appendix

\section{Proof of Theorem 5.1}

We start by restating the main result of \cite{LePage},  in 
a somewhat more explicit form. Recall that $f:\P_n\to\reals$
is said to be $\lambda$-H\"older continuous if its \emph{H\"older 
coefficient}, defined by
\begin{eqnarray}
[f]_{\lambda} = \sup_{\bx\neq \by}\frac{|f(\bx)-f(\by)|}{d(\bx,\by)^{\lambda}}\, ,
\end{eqnarray}
is finite.
\begin{theorem}[Le Page, 1982]\label{thm:Page}
Under assumptions \LA\, and \LB\, there exists a unique 
measure $\mu$ on $\P_n$ that is stationary for the Markov chain
$\{\bX_t\}$. Further, there exists constants $A\ge 0$, $\rho\in (0,1)$,
$\lambda\in (0,1]$
such that, for any $\lambda$-H\"older function $f:\P_n\to\reals$,
\begin{eqnarray*}
\big|\E\{f(\bX_t)\}-\mu(f)\big|\le \, A\, \rho^t\, [f]_{\lambda}\, .
\end{eqnarray*}
\end{theorem}
{\bf Remark: } The above follows immediately from Theorem 1
in \cite{LePage}
via a simple coupling argument. Notice in particular that it applies 
to any Lipschitz function since  $[f]_{\lambda}$
is upper bounded by the Lipschitz modulus of $f$.

Next we restate and prove the first part of Theorem 
\ref{thm:muproperty}.
\begin{theorem}\label{thm:muproperty0}
Assume conditions $\LA$ and $\LB$ hold, together with 
$\Aa$, $\Ab$.  Denote by $\stat$ the unique stationary measure 
of the Markov chain $\{\bX_t\}_{t\ge 0}$.
Then 
\begin{eqnarray}
\stat(\G^c) = 0.
\end{eqnarray}
\end{theorem}
\begin{proof}
Consider a Markov chain $MC_1$ with $\bx_0 \in \G$. The Markov chain $MC_1$ has a
stationary distribution because conditions $\LA$ and $\LB$ hold. From the
property \Aa, we know that $\bx_t\in \G$. Therefore the stationary distribution of
$MC_1$, say $\mu_1$, satisfies $\mu_1(\G^c) =0$. From Theorem~\ref{thm:Page} we
know that the stationary distribution is unique. Therefore $\mu = \mu_1$ and
hence $\mu(\G^c) = 0$.
\end{proof}

Finally, we will state and prove a generalization of
the second part of Theorem \ref{thm:muproperty}. For 
this we generalize hypothesis $\Ab$ as follows.
\begin{enumerate}
\item[\Abb.] For any $\bx\neq\by\in \G$, $\E\left[d(\Mt[t]\bx,\Mt[t]\by)
^{\lambda}
\right] \le \rho \, d(\bx,\by)^{\lambda}$.
\end{enumerate}

\begin{theorem}\label{thm:GeneralConv}
Assume conditions \LA\, and \LB\, hold, together with 
\Aa\, and \Abb.  Denote by $\stat$ the unique stationary measure 
of the Markov chain $\{\bX_t\}_{t\ge 0}$. Let $\bx_{0}\in \G$.
Then for any $\lambda$-H\"older function $f:\P_n\to\reals$,
we have
\begin{eqnarray}
\big|\E\{f(\bX_t)\}-\stat(f)\big|\le \, \rho^t\, [f]_{\lambda}\, .
\end{eqnarray}
%
%and 
%\begin{eqnarray}
%\E\big\{(f(\bX_t)-\nu(f))^2\big\}& \leq \rho_{2\lambda}^t[f]_{\lambda}^2  \,.
%\end{eqnarray}
\end{theorem}
The proof of this theorem is based on a coupling 
argument. The coupling assumed throughout is fairly simple:
given initial conditions $\bx_0$, $\by_0\in\G$, we define
the chain $\{(\bX_t,\bY_t)\}_{t\ge 0}$ by letting
$(\bX_0,\bY_0) = (\bx_0,\by_0)$ and, for all $t\ge 1$,
\begin{eqnarray}
\bX_t = \Mt[t]\Mt[t-1]\cdots \Mt[1]\bx_0\, ,\;\;\;\;\;\;\;
\bY_t = \Mt[t]\Mt[t-1]\cdots \Mt[1]\by_0\, .
\end{eqnarray}
It is further convenient to introduce,
for $t\in\naturals$, $\lambda > 0$ the quantity
\begin{eqnarray}
\rho_{\lambda}(t) \equiv \sup_{\bx_0\neq\by_0\in \G}\E\Big\{
\Big[\frac{d(\bX_t,\bY_t)}{d(\bX_0,\bY_0)}\Big]^{\lambda}\Big\}\, .
\end{eqnarray}
\begin{proof} %[Proof of Theorem \ref{thm:GeneralConv}]
First notice that the function  
$t\mapsto\rho_{\lambda}(t)$ is submultiplicative. This follows from
\begin{align*}
\rho_{\lambda}&(t_1+t_2) = \sup_{\bx_0\neq\by_0\in \G}\E\Big\{
\Big[\frac{d(\bX_{t_1+t_2},\bY_{t_1+t_2})}{d(\bX_0,\bY_0)}\Big]^{\lambda}\Big\} \nonumber \\
&=  \sup_{\bx_0\neq\by_0\in \G}\E\Big\{
\Big[\frac{d(\bX_{t_1},\bY_{t_1})}{d(\bX_0,\bY_0)}\Big]^{\lambda}
\Big[\frac{d(\bX_{t_1+t_2},\bY_{t_1+t_2})}{d(\bX_{t_1},\bY_{t_1})}\Big]^{\lambda}\Big\}\nonumber\\
&\stackrel{(a)}{\le} \sup_{\bx_0\neq\by_0\in \G}\E\Big\{
\Big[\frac{d(\bX_{t_1},\bY_{t_1})}{d(\bX_0,\bY_0)}\Big]^{\lambda}\Big\}\sup_{\bx_0\neq\by_0\in
\G}\E\Big\{
\Big[\frac{d(\bX_{t_2},\bY_{t_2})}{d(\bX_0,\bY_0)}\Big]^{\lambda}\Big\} \nonumber \\
&= \rho_{\lambda}(t_1)\rho_{\lambda}(t_2)\, .
\end{align*}
where $(a)$ follows from the condition \Aa. From the condition \Abb, we know that
$\rho_\lambda(1)\leq\rho$, hence  
$\rho_{\lambda}(t) \le \rho^t$.

Next let $\{\bX_{t}\}_{t\ge 0}$ and $\{\bY_{t}\}_{t\ge 0}$ be
Markov chains coupled as above, with initial conditions
$\bX_0=\bx_0 \in \G$ and $\bY_0\sim\mu$. We then have
\begin{align*}
\big|\E\{f(\bX_t)\}-\mu(f)\big|&= 
\big|\E\{f(\bX_t)\}-\E\{f(\bY_t)\}\big|\\
&\le
\E\big\{|f(\bX_t)-f(\bY_t)|\big\}\\
&\le [f]_{\lambda}
\E\big\{d(\bX_t,\bY_t)^{\lambda}\big\}\\
&\le [f]_{\lambda}
\E\Big\{\Big[\frac{d(\bX_t,\bY_t)}{d(\bX_0,\bY_0)}\Big]^{\lambda}
\Big\}\\
&\le  [f]_{\lambda}\, \rho_{\lambda}(t) \leq [f]_\lambda \rho^t\, ,
\end{align*}
where we used $d(\bX_0,\bY_0)\le 1$ by definition. 
This concludes our proof.
%
%Similarly,
%\begin{align*}
%%
%\E\big\{(f(\bX_t)-\nu(f))^2\big\}&= 
%\E\big\{(f(\bX_t)-\bE\{f(\bY_t)\})^2\big\}
%\le  \E\big\{(f(\bX_t)-f(\bY_t))^2\big\}\\
%&\le [f]_{\lambda}^2
%\E\big\{d(\bX_t,\bY_t)^{2\lambda}\big\}
%\le [f]_{\lambda}^2
%\E\Big\{\Big[\frac{d(\bX_t,\bY_t)}{d(\bX_0,\bY_0)}\Big]^{2\lambda}
%\Big\}\le  [f]_{\lambda}^2\, \rho_{2\lambda}(t) \leq [f]_{\lambda}^2 \rho_{2\lambda}^{t}\,.
%%
%\end{align*}
\end{proof}

% That's all folks!
\end{document}